\documentclass[12pt,letterpaper]{article}
\usepackage{epsfig,rotating,setspace,latexsym,amsmath,epsf,amssymb,bm}
\usepackage{theorem}
\usepackage{cite,appendix}

\title{Secrecy Capacity of a Class of Broadcast Channels with an Eavesdropper\thanks{This work was supported by NSF
Grants CCF 04-47613, CCF 05-14846, CNS 07-16311 and CCF 07-29127,
and was presented in part at the 42nd Asilomar Conference on
Signals, Systems and Computers, Pacific Grove, CA, October 2008
\cite{Asilomar_08}.}}

\author{Ersen Ekrem \qquad Sennur Ulukus \\
\normalsize Department of Electrical and Computer Engineering\\
\normalsize University of Maryland, College Park, MD 20742 \\
\normalsize {\it ersen@umd.edu} \qquad {\it ulukus@umd.edu} }

\newcommand{\bu}{{\mathbf{u}}}

\newcommand{\bx}{{\mathbf{x}}}

\newcommand{\by}{{\mathbf{y}}}

\theorembodyfont{\normalfont}
\newtheorem{Remark}{Remark}
\newtheorem{Theo}{Theorem}
\newtheorem{Prop}{Proposition}
\newtheorem{Lem}{Lemma}
\newtheorem{Cor}{Corollary}

\newenvironment{proof}[1]{\medskip\par\noindent
{\bf Proof:\,}\,#1}{{\mbox{\,$\blacksquare$}\par}}
\begin{document}
\date{}
\maketitle

\setstretch{1.1}

\begin{abstract}
We study the security of communication between a single
transmitter and multiple receivers in a broadcast channel in the
presence of an eavesdropper. Characterizing the secrecy capacity
region of this channel in its most general form is difficult,
because the version of this problem without any secrecy
constraints, is the broadcast channel with an arbitrary number of
receivers, whose capacity region is open. Consequently, to have
progress in understanding secure broadcasting, we resort to
studying several special classes of channels, with increasing
generality. As the first model, we consider the degraded
multi-receiver wiretap channel where the legitimate receivers
exhibit a degradedness order while the eavesdropper is more noisy
with respect to all legitimate receivers. We establish the secrecy
capacity region of this channel model. Secondly, we consider the
parallel multi-receiver wiretap channel with a less noisiness
order in each sub-channel, where this order is not necessarily the
same for all sub-channels. Consequently, this parallel
multi-receiver wiretap channel is not as restrictive as the
degraded multi-receiver wiretap channel, because the overall
channel does not exhibit a degradedness or even a less noisiness
order. We establish the common message secrecy capacity and sum
secrecy capacity of this channel. Thirdly, we study a special
class of parallel multi-receiver wiretap channels and provide a
stronger result. In particular, we study the case with two
sub-channels two users and one eavesdropper, where there is a
degradedness order in each sub-channel such that in the first
(resp. second) sub-channel the second (resp. first) receiver is
degraded with respect to the first (resp. second) receiver, while
the eavesdropper is degraded with respect to both legitimate
receivers in both sub-channels. We determine the secrecy capacity
region of this channel, and discuss its extensions to arbitrary
numbers of users and sub-channels. Finally, we focus on a variant
of this previous channel model where the transmitter can use only
one of the sub-channels at any time. We characterize the secrecy
capacity region of this channel as well.
\end{abstract}

\setstretch{1.2}

\newpage
\section{Introduction}
Information theoretic secrecy was initiated by Wyner in his
seminal work \cite{Wyner} where he introduced the wiretap channel
and established the capacity-equivocation region of the {\it
degraded} wiretap channel. Later, his result was generalized to
arbitrary, {\it not necessarily degraded}, wiretap channels by
Csiszar and Korner \cite{Korner}. Recently, many multiuser channel
models have been considered from a secrecy point of view
\cite{Aylin_2,Aylin_Cooperative,Allerton_08,Ruoheng,Ruoheng2,Broadcasting_Wornell,Khandani_Degraded_Wiretap,
Oohama, Hesham,Melda_1,He_journal,Aylin_Yener,CISS_08,ISIT_08,
Bloch_relay,Yingbin_1,Ruoheng_3,Ruoheng_4}. One basic extension of
the wiretap channel to the multiuser environment is {\it secure
broadcasting to many users} in the presence of an eavesdropper. In
the most general form of this problem (see
Figure~\ref{fig_secure_broadcasting}), one transmitter wants to
have confidential communication with an arbitrary number of users
in a broadcast channel, while this communication is being
eavesdropped by an external entity. Our goal is to understand the
theoretical limits of secure broadcasting, i.e., largest
simultaneously achievable secure rates. Characterizing the secrecy
capacity region of this channel model in its most general form is
difficult, because the version of this problem without any secrecy
constraints, is the broadcast channel with an arbitrary number of
receivers, whose capacity region is open. Consequently, to have
progress in understanding the limits of secure broadcasting, we
resort to studying several special classes of channels, with
increasing generality. The approach of studying special channel
structures was also followed in the existing literature on secure
broadcasting \cite{Khandani_Degraded_Wiretap,
Broadcasting_Wornell}.

Reference \cite{Khandani_Degraded_Wiretap} first considers an
arbitrary wiretap channel with two legitimate receivers and one
eavesdropper, and provides an inner bound for achievable rates
when each user wishes to receive an independent message. Secondly,
\cite{Khandani_Degraded_Wiretap} focuses on the degraded wiretap
channel with two receivers and one eavesdropper, where there is a
degradedness order among the receivers, and the eavesdropper is
degraded with respect to both users (see
Figure~\ref{fig_less_noisy} for a more general version of the
problem that we study). For this setting,
\cite{Khandani_Degraded_Wiretap} finds the secrecy capacity
region. This result is concurrently and independently obtained in
this work as a special case, see
Corollary~\ref{degraded_multiuser_wiretap}, which is also
published in a conference version in \cite{Asilomar_08}.

Another relevant work on secure broadcasting is
\cite{Broadcasting_Wornell} which considers secure broadcasting to
$K$ users using $M$ sub-channels (see Figure~\ref{fig_parallel})
for two different scenarios: In the first scenario, the
transmitter wants to convey only a common confidential message to
all users, and in the second scenario, the transmitter wants to
send independent messages to all users. For both scenarios,
\cite{Broadcasting_Wornell} considers a sub-class of parallel
multi-receiver wiretap channels, where in any given sub-channel
there is a degradation order such that each receiver's observation
(except the best one) is a degraded version of some other
receiver's observation, and this degradation order is not
necessarily the same for all sub-channels. For the first scenario,
\cite{Broadcasting_Wornell} finds the common message secrecy
capacity for this sub-class. For the second scenario, where each
user wishes to receive an independent message,
\cite{Broadcasting_Wornell} finds the sum secrecy capacity for
this sub-class of channels.

\begin{figure}[t]
\begin{center}
\includegraphics[width=6.5cm]{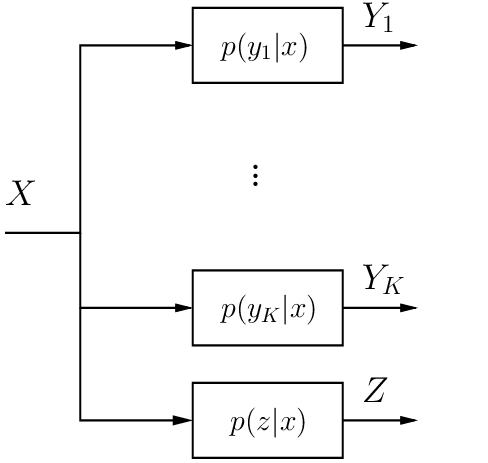}
\caption{Secure broadcasting to many users in the presence of an
eavesdropper.} \label{fig_secure_broadcasting}
\end{center}
\end{figure}

\begin{figure}[htp]
\begin{center}
\epsfxsize = 12 cm  \epsffile{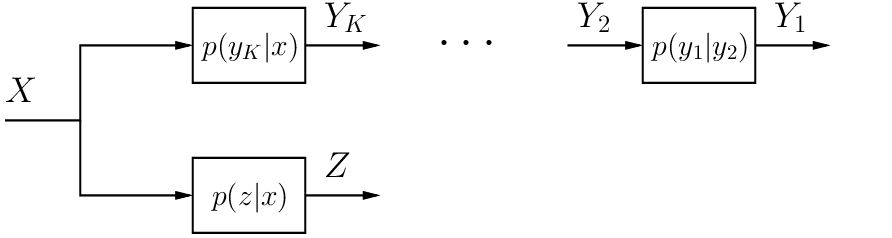}
\end{center}
\caption{The degraded multi-receiver wiretap channel with a more
noisy eavesdropper.} \label{fig_less_noisy}
\end{figure}

\begin{figure}[htp]
\begin{center}
\includegraphics[width=13cm]{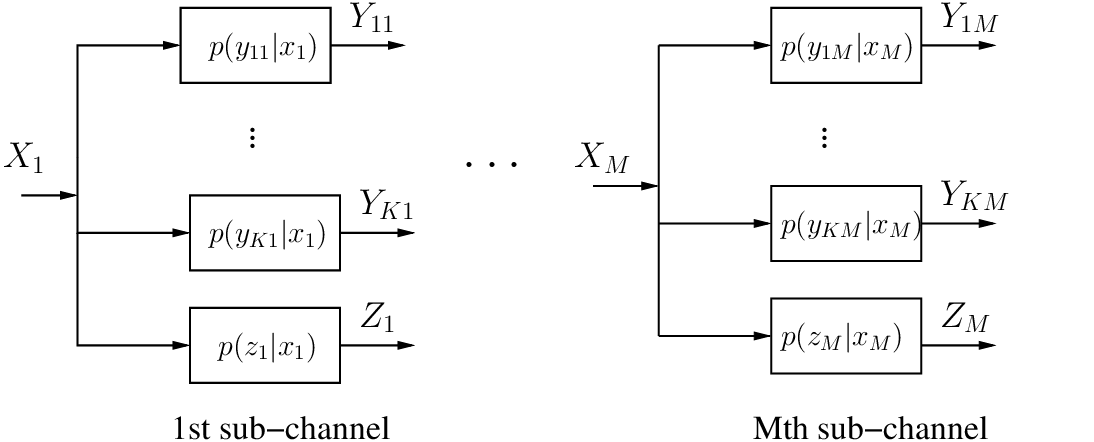}
\end{center}
\caption{The parallel multi-receiver wiretap channel.}
\label{fig_parallel}
\end{figure}

In this paper, our approach will be two-fold: First, we will
identify more general channel models than considered in
\cite{Khandani_Degraded_Wiretap, Broadcasting_Wornell} and
generalize the results in \cite{Khandani_Degraded_Wiretap,
Broadcasting_Wornell} to those channel models, and secondly, we
will consider somewhat more specialized channel models than in
\cite{Broadcasting_Wornell} and provide more comprehensive
results. More precisely, our contributions in this paper are:

\begin{enumerate}
\item We consider the degraded multi-receiver wiretap channel with
an arbitrary number of users and one eavesdropper, where users are
arranged according to a degradedness order, and each user has a
less noisy channel with respect to the eavesdropper, see
Figure~\ref{fig_less_noisy}. We find the secrecy capacity region
when each user receives both an independent message and a common
message. Since degradedness implies less noisiness \cite{Korner},
this channel model contains the sub-class of channel models where
in addition to the degradedness order users exhibit, the
eavesdropper is degraded with respect to all users. Consequently,
our result can be specialized to the degraded multi-receiver
wiretap channel with an arbitrary number of users and a degraded
eavesdropper, see Corollary~\ref{degraded_multiuser_wiretap} and
also \cite{Asilomar_08}. The two-user version of the degraded
multi-receiver wiretap channel was studied and the capacity region
was found independently and concurrently in
\cite{Khandani_Degraded_Wiretap}.

\item We then focus on a class of parallel multi-receiver wiretap
channels with an arbitrary number of legitimate receivers and an
eavesdropper, see Figure~\ref{fig_parallel}, where in each
sub-channel, for any given user, either the user's channel is less
noisy with respect to the eavesdropper's channel, or vice versa.
We establish the common message secrecy capacity of this channel,
which is a generalization of the corresponding capacity result in
\cite{Broadcasting_Wornell} to a broader class of channels.
Secondly, we study the scenario where each legitimate receiver
wishes to receive an independent message for another sub-class of
parallel multi-receiver wiretap channels. For channels belonging
to this sub-class, in each sub-channel, there is a less noisiness
order which is not necessarily the same for all sub-channels.
Consequently, this ordered class of channels is a subset of the
class for which we establish the common message secrecy capacity.
We find the sum secrecy capacity for this class, which is again a
generalization of the corresponding result in
\cite{Broadcasting_Wornell} to a broader class of channels.

\item We also investigate a class of parallel multi-receiver
wiretap channels with two sub-channels, two users and one
eavesdropper, see Figure~\ref{fig_unmatched}. For the channels in
this class, there is a specific degradation order in each
sub-channel such that in the first (resp. second) sub-channel the
second (resp. first) user is degraded with respect to the first
(resp. second) user, while the eavesdropper is degraded with
respect to both users in both sub-channels. This is the model of
\cite{Broadcasting_Wornell} for $K=2$ users and $M=2$
sub-channels. This model is more restrictive compared to the one
mentioned in the previous item. Our motivation to study this more
special class is to provide a stronger and more comprehensive
result. In particular, for this class, we determine the entire
secrecy capacity region when each user receives both an
independent message and a common message. In contrast,
\cite{Broadcasting_Wornell} gives the common message secrecy
capacity (when only a common message is transmitted) and sum
secrecy capacity (when only independent messages are transmitted)
of this class. We discuss the generalization of this result to
arbitrary numbers of users and sub-channels.

\item We finally consider a variant of the previous channel model.
In this model, we again have a parallel multi-receiver wiretap
channel with two sub-channels, two users and one eavesdropper, and
the degradation order in each sub-channel is exactly the same as
in the previous item. However, in this case, the input and output
alphabets of one sub-channel are non-intersecting with the input
and output alphabets of the other sub-channel. Moreover, we can
use only one of these sub-channels at any time. We determine the
secrecy capacity region of this channel when the transmitter sends
both an independent message to each receiver and a common message
to both receivers.
\end{enumerate}

\begin{figure}[t]
\begin{center}
\epsfxsize = 11 cm  \epsffile{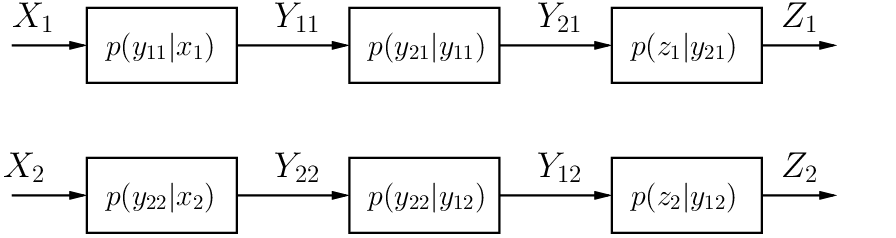}
\end{center}
\caption{The parallel degraded multi-receiver wiretap channel.}
\label{fig_unmatched}
\end{figure}

\section{Degraded Multi-receiver Wiretap Channels}
\label{sec:degraded_mr_wt}

We first consider the generalization of Wyner's degraded wiretap
channel to the case with many legitimate receivers. In particular,
the channel consists of a transmitter with an input alphabet
$x\in\mathcal{X}$, $K$ legitimate receivers with output alphabets
$y_k\in\mathcal{Y}_k,~k=1,\ldots,K,$ and an eavesdropper with
output alphabet $z\in\mathcal{Z}$. The transmitter sends a
confidential message to each user, say $w_k\in\mathcal{W}_k$ to
the $k$th user, in addition to a common message,
$w_0\in\mathcal{W}_0$, which is to be delivered to all users. All
messages are to be kept secret from the eavesdropper. The channel
is assumed to be memoryless with a transition probability
$p(y_1,y_2,\ldots,y_K,z|x)$.

In this section, we consider a special class of these channels,
see Figure~\ref{fig_less_noisy}, where users exhibit a certain
degradation order, i.e., their channel outputs satisfy the
following Markov chain
\begin{align}
X\rightarrow Y_K\rightarrow\ldots\rightarrow Y_1
\end{align}
and each user has a less noisy channel with respect to the
eavesdropper, i.e., we have
\begin{align}
I(U;Y_k)>I(U;Z) \label{less_noisy_wt}
\end{align}
for every $U$ such that $U\rightarrow X\rightarrow (Y_k,Z)$. In
fact, since a degradation order exists among the users, it is
sufficient to say that user 1 has a less noisy channel with
respect to the eavesdropper to guarantee that all users do.
Hereafter, we call this channel {\it the degraded multi-receiver
wiretap channel with a more noisy eavesdropper}. We note that this
channel model contains the degraded multi-receiver wiretap channel
which is defined through the Markov chain
\begin{align}
X\rightarrow Y_K\rightarrow\ldots\rightarrow Y_1\rightarrow Z
\label{degraded_wc}
\end{align}
because the Markov chain in (\ref{degraded_wc}) implies the less
noisiness condition in (\ref{less_noisy_wt}).

A $(2^{nR_0},2^{nR_1},\ldots,2^{nR_K},n)$ code for this channel
consists of $K+1$ message sets,
$\mathcal{W}_k=\{1,\ldots,2^{nR_k}\}$, $k=0,1,\ldots,K$, an
encoder $f:\mathcal{W}_0\times\ldots\times\mathcal{W}_K\rightarrow
\mathcal{X}^{n}$, $K$ decoders, one at each legitimate receiver,
$g_k:\mathcal{Y}_k\rightarrow \mathcal{W}_0\times \mathcal{W}_k$,
$k=1,\ldots,K$. The probability of error is defined as
$P_e^n=\max_{k=1,\ldots,K}\Pr\left[g_k(Y_k^n)\neq
(W_0,W_k)\right]$. A rate tuple $(R_0,R_1,\ldots,R_K)$ is said to
be achievable if there exists a code with
$\lim_{n\rightarrow\infty}P_e^n=0$ and
\begin{align}
\lim_{n\rightarrow\infty}\frac{1}{n}H(\mathcal{S}(W)|Z^n)\geq
\sum_{k\in \mathcal{S}(W)}R_k, \quad \forall ~\mathcal{S}(W)
\label{perfect_secrecy}
\end{align}
where $\mathcal{S}(W)$ denotes any subset of
$\{W_0,W_1,\ldots,W_K\}$. Hence, we consider only perfect secrecy
rates. The secrecy capacity region is defined as the closure of
all achievable rate tuples.

The secrecy capacity region of the degraded multi-receiver wiretap
channel with a more noisy eavesdropper is given by the following
theorem whose proof is provided in
Appendix~\ref{proof_of_less_noisy_multiuser_wiretap}.
\begin{Theo}
\label{less_noisy_nultiuser_wiretap} The secrecy capacity region
of the degraded multi-receiver wiretap channel with a more noisy
eavesdropper is given by the union of the rate tuples
$(R_0,R_1,\ldots,R_K)$ satisfying
\begin{align}
R_0+\sum_{k=1}^\ell R_k & \leq  \sum_{k=1}^\ell
I(U_k;Y_k|U_{k-1})-I(U_{\ell};Z), \quad \ell=1,\ldots,K
\label{capacity_less_noisy}
\end{align}
where $U_0=\phi, U_K=X$, and the union is over all probability
distributions of the form
\begin{align}
p(u_1)p(u_2|u_1)\ldots p(u_{K-1}|u_{K-2})p(x|u_{K-1})
\end{align}
\end{Theo}

\begin{Remark}
Theorem~\ref{less_noisy_nultiuser_wiretap} implies that a modified
version of superposition coding can achieve the boundary of the
capacity region. The difference between the superposition coding
scheme used to achieve (\ref{capacity_less_noisy}) and the
standard one \cite{cover_book} is that the former uses stochastic
encoding in each layer of the code to associate each message with
many codewords. This controlled amount of redundancy prevents the
eavesdropper from being able decode the message.
\end{Remark}

As stated earlier, the degraded multi-receiver wiretap channel
with a more noisy eavesdropper contains the degraded
multi-receiver wiretap channel which requires the eavesdropper to
be degraded with respect to all users as stated
(\ref{degraded_wc}). Thus, we can specialize our result in
Theorem~\ref{less_noisy_nultiuser_wiretap} to the degraded
multi-receiver wiretap channel as given in the following corollary
whose proof is provided in
Appendix~\ref{proof_of_corollary_degraded_multiuser_wiretap}.

\begin{Cor}
\label{degraded_multiuser_wiretap} The secrecy capacity region of
the degraded multi-receiver wiretap channel is given by the union
of the rate tuples $(R_0,R_1,\ldots,R_K)$ satisfying
\begin{align}
R_0+\sum_{k=1}^\ell R_k & \leq  \sum_{k=1}^\ell
I(U_k;Y_k|U_{k-1},Z), \quad \ell=1,\ldots,K
\end{align}
where $U_0=\phi,U_K=X$, and the union is over all probability
distributions of the form
\begin{align}
p(u_1)p(u_2|u_1)\ldots
p(u_{K-1}|u_{K-2})p(x|u_{K-1})
\end{align}
\end{Cor}

We acknowledge an independent and concurrent work regarding the
degraded multi-receiver wiretap channel. Reference
\cite{Khandani_Degraded_Wiretap} considers the two-user case and
establishes the secrecy capacity region as well.

So far we have determined the entire secrecy capacity region of
the degraded multi-receiver wiretap channel with a more noisy
eavesdropper. This class of channels requires a certain
degradation order among the legitimate receivers which may be
viewed as being too restrictive from a practical point of view.
Our goal is to consider progressively more general channel models.
Towards that goal, in the next section, we consider channel models
where the users are not ordered in a degradedness or noisiness
order. However, the concepts of degradedness and noisiness are
essential in proving capacity results. In the next section, we
will consider multi-receiver broadcast channels which are composed
of independent sub-channels. We will assume some noisiness
properties in these sub-channels in order to derive certain
capacity results. However, even though the sub-channels will have
certain noisiness properties, the overall broadcast channel will
not have any degradedness or noisiness properties.

\section{Parallel Multi-receiver Wiretap Channels}
\label{sec:parallel}

Here, we investigate the parallel multi-receiver wiretap channel
where the transmitter communicates with $K$ legitimate receivers
using $M$ independent sub-channels in the presence of an
eavesdropper, see Figure~\ref{fig_parallel}. The channel
transition probability of a parallel multi-receiver wiretap
channel is
\begin{align}
p\left(\left\{y_{1m},\ldots,y_{Km},z_m\right\}_{m=1}^{M}|\left\{x_m\right\}_{m=1}^{M}\right)=\prod_{m=1}^M
p\left(y_{1m},\ldots,y_{Km},z_m|x_m\right)
\label{channel_transition}
\end{align}
where $x_{m}\in\mathcal{X}_m$ is the input in the $m$th
sub-channel where $\mathcal{X}_m$ is the corresponding channel
input alphabet, $y_{km}\in\mathcal{Y}_{km}$ (resp.
$z_{m}\in\mathcal{Z}_{m}$) is the output in the $k$th user's
(resp. eavesdropper's) $m$th sub-channel where $\mathcal{Y}_{km}$
(resp. $\mathcal{Z}_m$) is the $k$th user's (resp. eavesdropper's)
$m$th sub-channel output alphabet.

In this section, we investigate special classes of parallel
multi-receiver wiretap channels. These channel models contain the
class of channel models studied in \cite{Broadcasting_Wornell} as
a special case. Similar to \cite{Broadcasting_Wornell}, our
emphasis will be on the common message secrecy capacity and the
sum secrecy capacity.

\subsection{The Common Message Secrecy Capacity}
\label{sec:parallel_com_message}

We first consider the simplest possible scenario where the
transmitter sends a common confidential message to all users.
Despite its simplicity, the secrecy capacity of a common
confidential message (hereafter will be called the common message
secrecy capacity) in a general broadcast channel is unknown.

The common message secrecy capacity for a special class of
parallel multi-receiver wiretap channels was studied in
\cite{Broadcasting_Wornell}. In this class of parallel
multi-receiver wiretap channels \cite{Broadcasting_Wornell}, each
sub-channel exhibits a certain degradation order which is not
necessarily the same for all sub-channels, i.e., the following
Markov chain is satisfied
\begin{align}
X_{l}\rightarrow Y_{\pi_{l}(1)}\rightarrow
Y_{\pi_{l}(2)}\rightarrow \ldots \rightarrow Y_{\pi_{l}(K+1)}
\label{degradedness_order}
\end{align}
in the $l$th sub-channel, where
$(Y_{\pi_{l}(1)},Y_{\pi_{l}(2)},\ldots,Y_{\pi_{l}(K+1)})$ is a
permutation of $(Y_{1l},\ldots,Y_{Kl},Z_l)$. Hereafter, we call
this channel the parallel degraded multi-receiver wiretap
channel\footnote{In \cite{Broadcasting_Wornell}, these channels
are called {\it reversely degraded} parallel channels. Here, we
call them parallel degraded multi-receiver wiretap channels to be
consistent with the terminology used in the rest of the paper.}.
Although \cite{Broadcasting_Wornell} established the common
message secrecy capacity for this class of channels, in fact,
their result is valid for the broader class in which we have
either
\begin{align}
X_l\rightarrow Y_{kl}\rightarrow Z_l\label{degraded_1}
\end{align}
or
\begin{align}
X_l\rightarrow Z_l\rightarrow Y_{kl} \label{degraded_2}
\end{align}
valid for every $X_l$ and for any $(k,l)$ pair where
$k\in\left\{1,\ldots,K\right\}$, $l\in\left\{1,\ldots,M\right\}$.
Thus, it is sufficient to have a degradedness order between each
user and the eavesdropper in any sub-channel instead of the long
Markov chain between all users and the eavesdropper as in
(\ref{degradedness_order}).

Here, we focus on a broader class of channels where in each
sub-channel, for any given user, either the user's channel is less
noisy than the eavesdropper's channel, or vice versa. More
formally, we have either
\begin{align}
I(U;Y_{kl})>I(U;Z_l) \label{less_noisy_1}
\end{align}
or
\begin{align}
I(U;Y_{kl})<I(U;Z_l) \label{less_noisy_2}
\end{align}
for all $U\rightarrow X_l \rightarrow (Y_{kl},Z)$ and any $(k,l)$
pair where $k\in\left\{1,\ldots,K\right\}$,
$l\in\left\{1,\ldots,M\right\}$. Hereafter, we call this channel
{\em the parallel multi-receiver wiretap channel with a more noisy
eavesdropper}. Since the Markov chain in
(\ref{degradedness_order}) implies either (\ref{less_noisy_1}) or
(\ref{less_noisy_2}), the parallel multi-receiver wiretap channel
with a more noisy eavesdropper contains the parallel degraded
multi-receiver wiretap channel studied in
\cite{Broadcasting_Wornell}.

A $(2^{nR},n)$ code for this channel consists of a message set,
$\mathcal{W}_0=\{1,\ldots,2^{nR}\}$, an encoder, $f:
\mathcal{W}_0\rightarrow \mathcal{X}_1^n\times
\ldots\times\mathcal{X}_M^n$, $K$ decoders, one at each legitimate
receiver
$g_k:\mathcal{Y}_{k1}\times\ldots\times\mathcal{Y}_{kM}\rightarrow
\mathcal{W}_0, k=1,\ldots,K$. The probability of error is defined
as $P_e^n=\max_{k=1,\ldots,K} \Pr\left[\hat{W}_{k0}\neq
W_0\right]$ where $\hat{W}_{k0}$ is the $k$th user's decoder
output. The secrecy of the common message is measured through the
equivocation rate which is defined as
$\frac{1}{n}H(W_0|Z_1^n,\ldots,Z_M^n)$. A common message secrecy
rate, $R$, is said to be achievable if there exists a code such
that $\lim_{n\rightarrow\infty}P_e^n=0$, and
\begin{align}
\lim_{n\rightarrow\infty}\frac{1}{n}H(W_0|Z_1^n,\ldots,Z_M^n)\geq
R
\end{align}
The common message secrecy capacity is the supremum of all
achievable secrecy rates.

The common message secrecy capacity of the parallel multi-receiver
wiretap channel with a more noisy eavesdropper is stated in the
following theorem whose proof is given in
Appendix~\ref{proof_of_common_less_noisier_parallel}.
\begin{Theo}
\label{common_less_noisier_parallel} The common message secrecy
capacity, $C_0$, of the parallel multi-receiver wiretap channel
with a more noisy eavesdropper is given by
\begin{align}
C_0&= \max \min_{k=1,\ldots,K}\sum_{l=1}^M
\big[I(X_l;Y_{kl})-I(X_l;Z_l)\big]^+ \label{C_PLNMWC}
\end{align}
where the maximization is over all distributions of the form $
p(x_1,\ldots,x_M)=\prod_{l=1}^M p(x_l) $.
\end{Theo}

\begin{Remark}
Theorem~\ref{common_less_noisier_parallel} implies that we should
not use the sub-channels in which there is no user that has a less
noisy channel than the eavesdropper. Moreover,
Theorem~\ref{common_less_noisier_parallel} shows that the use of
independent inputs in each sub-channel is sufficient to achieve
the capacity, i.e., inducing correlation between channel inputs of
sub-channels cannot provide any improvement.
\end{Remark}

As stated earlier, the parallel multi-receiver wiretap channel
with a more noisy eavesdropper encompasses the parallel degraded
multi-receiver wiretap channel studied in
\cite{Broadcasting_Wornell}. Hence, we can specialize
Theorem~\ref{common_less_noisier_parallel} to recover the common
message secrecy capacity of the parallel degraded multi-receiver
wiretap channel established in \cite{Broadcasting_Wornell}. This
is stated in the following corollary whose proof is given in
Appendix~\ref{proof_of_corollary_degraded_common}.
\begin{Cor}
\label{corollary_degraded_common} The common message secrecy
capacity of the  parallel degraded multi-receiver wiretap channel
is given by
\begin{align}
C_0&= \max \min_{k=1,\ldots,K}\sum_{l=1}^M I(X_l;Y_{kl}|Z_l)
\label{C_PDMWC}
\end{align}
where the maximization is over all distributions of the form $
p(x_1,\ldots,x_M)=\prod_{l=1}^M p(x_l)$.
\end{Cor}

\subsection{The Sum Secrecy Capacity}
\label{subsec:parallel}

We now consider the scenario where the transmitter sends an
independent confidential message to each legitimate receiver, and
focus on the sum secrecy capacity. We consider a class of parallel
multi-receiver wiretap channels where the legitimate receivers and
the eavesdropper exhibit a certain less noisiness order in each
sub-channel. These less noisiness orders are not necessarily the
same for all sub-channels. Therefore, the overall channel does not
have a less noisiness order. In the $l$th sub-channel, for all
$U\rightarrow X_l\rightarrow (Y_{1l},\ldots,Y_{Kl},Z_l)$, we have
\begin{align}
I(U;Y_{\pi_{l}(1)})>I(U;Y_{\pi_{l}(2)})>\ldots>I(U;Y_{\pi_{l}(K+1)})
\end{align}
where  $(Y_{\pi_{l}(1)},Y_{\pi_{l}(2)},\ldots,Y_{\pi_{l}(K+1)})$
is a permutation of $(Y_{1l},\ldots,Y_{Kl},Z_l)$. We call this
channel {\em the parallel multi-receiver wiretap channel with a
less noisiness order in each sub-channel}. We note that this class
of channels is a subset of the parallel multi-receiver wiretap
channel with a more noisy eavesdropper studied in
Section~\ref{sec:parallel_com_message}, because of the additional
ordering imposed between users' sub-channels. We also note that
the class of parallel degraded multi-receiver wiretap channels
with a degradedness order in each sub-channel studied in
\cite{Broadcasting_Wornell} is not only a subset of parallel
multi-receiver wiretap channels with a more noisy eavesdropper
studied in Section~\ref{sec:parallel_com_message} but also a
subset of parallel multi-receiver wiretap channels with a less
noisiness order in each sub-channel studied in this section.

A $(2^{nR_1}\ldots,2^{nR_K},n)$ code for this channel consists of
$K$ message sets, $\mathcal{W}_k=\{1,\ldots,2^{nR_k}\}, \break
k=1,\ldots,K$, an encoder, $f:
\mathcal{W}_1\times\ldots\times\mathcal{W}_K\rightarrow
\mathcal{X}_1^n\times \ldots\times\mathcal{X}_M^n$, $K$ decoders,
one at each legitimate receiver
$g_k:\mathcal{Y}_{k1}\times\ldots\times\mathcal{Y}_{kM}\rightarrow
\mathcal{W}_k, k=1,\ldots,K$. The probability of error is defined
as $P_e^n=\max_{k=1,\ldots,K} \Pr\left[\hat{W}_{k}\neq W_k\right]$
where $\hat{W}_{k}$ is the $k$th user's decoder output. The
secrecy is measured through the equivocation rate which is defined
as $\frac{1}{n}H(W_1,\ldots,W_K|Z_1^n,\ldots,Z_M^n)$. A sum
secrecy rate, $R_s$, is said to be achievable if there exists a
code such that $\lim_{n\rightarrow\infty}P_e^n=0$, and
\begin{align}
\lim_{n\rightarrow\infty}\frac{1}{n}H(W_1,\ldots,W_K|Z_1^n,\ldots,Z_M^n)\geq
R_s
\end{align}
The sum secrecy capacity is defined to be the supremum of all
achievable sum secrecy rates.

The sum secrecy capacity for the class of parallel multi-receiver
wiretap channels with a less noisiness order in each sub-channel
studied in this section is stated in the following theorem whose
proof is given in Appendix~\ref{proof_of_sum_secrecy_capacity}.
\begin{Theo}
\label{sum_secrecy_capacity} The sum secrecy capacity of the
parallel multi-receiver wiretap channel with a less noisiness
order in each sub-channel is given by
\begin{align}
\max\sum_{l=1}^M \big[I(X_{l};Y_{\rho(l)l})-I(X_{l};Z_l)\big]^+
\end{align}
where the maximization is over all input distributions of the form
$ p(x_1,\ldots,x_M)=\prod_{l=1}^M p(x_l) $ and $\rho(l)$ denotes
the index of the strongest user in the $l$th sub-channel in the
sense that
\begin{align}
I(U;Y_{kl})\leq I(U;Y_{\rho(l)l})
\end{align}
for all  $U\rightarrow X_{l}\rightarrow
(Y_{1l},\ldots,Y_{Kl},Z_l)$ and any $k\in\{1,\ldots,K\}$.
\end{Theo}

\begin{Remark}
Theorem~\ref{sum_secrecy_capacity} implies that the sum secrecy
capacity is achieved by sending information only to the strongest
user in each sub-channel. As in
Theorem~\ref{common_less_noisier_parallel}, here also, the use of
independent inputs for each sub-channel is capacity-achieving.
\end{Remark}

As mentioned earlier, since the class of parallel multi-receiver
wiretap channels with a less noisiness order in each sub-channel
contains the class of parallel degraded multi-receiver wiretap
channels studied in \cite{Broadcasting_Wornell},
Theorem~\ref{sum_secrecy_capacity} can be specialized to give the
sum secrecy capacity of the latter class of channels as well. This
result was originally obtained in \cite{Broadcasting_Wornell}.
This is stated in the following corollary. Since the proof of this
corollary is similar to the proof of
Corollary~\ref{corollary_degraded_common}, we omit its proof.

\begin{Cor}
The sum secrecy capacity of the parallel degraded multi-receiver
wiretap channel is given by
\begin{align}
\max\sum_{l=1}^M I(X_{l};Y_{\rho(l)l}|Z_l)
\end{align}
where the maximization is over all input distributions of the form
$ p(x_1,\ldots,x_M)=\prod_{l=1}^M p(x_l) $ and $\rho(l)$ denotes
the index of the strongest user in the $l$th sub-channel in the
sense that
\begin{align}
X_l\rightarrow Y_{\rho(l)l}\rightarrow Y_{kl}
\end{align}
for all input distributions on $X_{l}$ and any
$k\in\{1,\ldots,K\}$.
\end{Cor}

So far, we have considered special classes of parallel
multi-receiver wiretap channels for specific scenarios and
obtained results similar to \cite{Broadcasting_Wornell}, only for
broader classes of channels. In particular, in
Section~\ref{sec:parallel_com_message}, we focused on the
transmission of a common message, whereas in
Section~\ref{subsec:parallel}, we focused on the sum secrecy
capacity when only independent messages are transmitted to all
users. In the subsequent sections, we will specialize our channel
model, but we will develop stronger and more comprehensive
results. In particular, we will let the transmitter send both
common and independent messages, and we will characterize the
entire secrecy capacity region.

\section{Parallel Degraded Multi-receiver Wiretap Channels}

We consider a special class of parallel degraded multi-receiver
wiretap channels with two sub-channels, two users and one
eavesdropper. We consider the most general scenario where each
user receives both an independent message and a common message.
All messages are to be kept secret from the eavesdropper.

For the special class of parallel degraded multi-receiver wiretap
channels in consideration, there is a specific degradation order
in each sub-channel. In particular, we have the following Markov
chain
\begin{align}
X_1\rightarrow Y_{11}  \rightarrow Y_{21}\rightarrow Z_1
\end{align}
in the first sub-channel, and the following Markov chain
\begin{align}
X_2\rightarrow Y_{22} \rightarrow Y_{12}\rightarrow Z_2
\end{align}
in the second sub-channel. Consequently, although in each
sub-channel, one user is degraded with respect to the other one,
this does not hold for the overall channel, and the overall
channel is not degraded for any user. The corresponding channel
transition probability is
\begin{align}
p(y_{11}|x_{1})p(y_{21}|y_{11})p(z_{1}|y_{21})p(y_{22}|x_{2})
p(y_{12}|y_{22})p(z_{2}|y_{12})
\label{channel_transition_degraded_two_users}
\end{align}
If we ignore the eavesdropper by setting $Z_1=Z_2=\phi$, this
channel model reduces to the broadcast channel that was studied in
\cite{Poltyrev, Product_Broadcast}.

A $(2^{nR_0},2^{nR_1},2^{nR_2},n)$ code for this channel consists
of three message sets, $\mathcal{W}_0=\{1,\ldots,\break
2^{nR_0}\}$, $\mathcal{W}_j=\{1,\ldots,2^{nR_j}\}, j=1,2$, one
encoder
$f:\mathcal{W}_0\times\mathcal{W}_1\times\mathcal{W}_2\rightarrow
\mathcal{X}_1^n\times\mathcal{X}_2^n$, two decoders one at each
legitimate receiver
$g_j:\mathcal{Y}_{j1}^n\times\mathcal{Y}_{j2}^n\rightarrow
\mathcal{W}_0\times\mathcal{W}_j, j=1,2$. The probability of error
is defined as $P_e^n=\max_{j=1,2}\Pr\left[g_j
(Y_{j1}^n,Y_{j2}^n)\neq (W_0,W_j)\right]$. A rate tuple
$(R_0,R_1,R_2)$ is said to be achievable if there exists a code
such that $\lim_{n\rightarrow\infty}P_e^n=0$ and
\begin{align}
\lim_{n\rightarrow\infty}\frac{1}{n}
H(\mathcal{S}(W)|Z_1^n,Z_2^n)\geq
\sum_{k\in\mathcal{S}(W)}R_k,\quad \forall ~\mathcal{S}(W)
\end{align}
where $\mathcal{S}(W)$ denotes any subset of $\{W_0,W_1,W_2\}$.
The secrecy capacity region is the closure of all achievable
secrecy rate tuples.

The secrecy capacity region of this parallel degraded
multi-receiver wiretap channel is characterized by the following
theorem whose proof is given in
Appendix~\ref{proof_of_product_wiretap_channel}.
\begin{Theo}
\label{Theorem_product_wiretap_channel} The secrecy capacity
region of the parallel degraded multi-receiver wiretap channel
defined by (\ref{channel_transition_degraded_two_users}) is the
union of the rate tuples $(R_0,R_1,R_2)$ satisfying
\begin{align}
R_0& \leq I(U_1;Y_{11}|Z_1)+I(U_2;Y_{12}|Z_2)\label{product_common_rate_1}\\
R_0& \leq I(U_1;Y_{21}|Z_1)+I(U_2;Y_{22}|Z_2)\label{product_common_rate_2}\\
R_0+R_1 & \leq I(X_1;Y_{11}|Z_1)+I(U_2;Y_{12}|Z_2)\label{product_common_private_rate_1}\\
R_0+R_2 & \leq I(X_2;Y_{22}|Z_2)+I(U_1;Y_{21}|Z_1)\label{product_common_private_rate_2}\\
R_0+R_1+R_2 & \leq
I(X_1;Y_{11}|Z_1)+I(U_2;Y_{12}|Z_2)+I(X_2;Y_{22}|U_2,Z_2)\label{product_sum_rate_1}\\
R_0+R_1+R_2 & \leq
I(X_2;Y_{22}|Z_2)+I(U_1;Y_{21}|Z_1)+I(X_1;Y_{11}|U_1,Z_1)
\label{product_sum_rate_2}
\end{align}
where the union is over all distributions of the form
$p(u_1,u_2,x_1,x_2)=p(u_1,x_1)p(u_2,x_2) $.
\end{Theo}

\begin{Remark}
If we let the encoder use an arbitrary joint distribution
$p(u_1,x_1,u_2,x_2)$ instead of the ones that satisfy
$p(u_1,x_1,u_2,x_2)=p(u_1,x_1)p(u_2,x_2)$, this would not enlarge
the region given in Theorem~\ref{Theorem_product_wiretap_channel},
because all rate expressions in
Theorem~\ref{Theorem_product_wiretap_channel} depend on either
$p(u_1,x_1)$ or $p(u_2,x_2)$ but not on the joint distribution
$p(u_1,u_2,x_1,x_2)$.
\end{Remark}

\begin{Remark}
The capacity achieving scheme uses either superposition coding in
both sub-channels or superposition coding in one of the
sub-channels, and a dedicated transmission in the other one. We
again note that this superposition coding is different from the
standard one~\cite{cover_book} in the sense that it associates
each message with many codewords by using stochastic encoding at
each layer of the code due to secrecy concerns.
\end{Remark}

\begin{Remark}
If we set $Z_1=Z_2=\phi$, we recover the capacity region of the
underlying broadcast channel~\cite{Product_Broadcast}.
\end{Remark}

\begin{Remark}
If we disable one of the sub-channels, say the first one, by
setting $Y_{11}=Y_{21}=Z_1=\phi$, the parallel degraded
multi-receiver wiretap channel of this section reduces to the
degraded multi-receiver wiretap channel of
Section~\ref{sec:degraded_mr_wt}. The corresponding secrecy
capacity region is then given by the union of the rate tuples
$(R_0,R_1,R_2)$ satisfying
\begin{align}
R_0+R_1&\leq I(U_2;Y_{12}|Z_2)\\
R_0+R_1+R_2&\leq I(X_2;Y_{22}|U_2,Z_2)+I(U_2;Y_{12}|Z_2)
\end{align}
where the union is over all $p(u_2,x_2)$. This region can be
obtained through either Corollary~\ref{degraded_multiuser_wiretap}
or Theorem~\ref{Theorem_product_wiretap_channel} (by setting
$Y_{11}=Y_{21}=Z_1=\phi$ and eliminating redundant bounds)
implying the consistency of the results.
\end{Remark}

Next, we consider the scenario where the transmitter does not send
a common message, and find the secrecy capacity region.
\begin{Cor}
The secrecy capacity region of the parallel degraded
multi-receiver wiretap channel defined through
(\ref{channel_transition_degraded_two_users}) with no common
message is given by the union of the rate pairs $(R_1,R_2)$
satisfying
\begin{align}
R_1 & \leq I(X_1;Y_{11}|Z_1)+I(U_2;Y_{12}|Z_2)\\
R_2 & \leq I(X_2;Y_{22}|Z_2)+I(U_1;Y_{21}|Z_1)\\
R_1+R_2 & \leq
I(X_1;Y_{11}|Z_1)+I(U_2;Y_{12}|Z_2)+I(X_2;Y_{22}|U_2,Z_2)\\
R_1+R_2 & \leq
I(X_2;Y_{22}|Z_2)+I(U_1;Y_{21}|Z_1)+I(X_1;Y_{11}|U_1,Z_1)
\end{align}
where the union is over all distributions of the form
$p(u_1)p(u_2)p(x_1|u_1)p(x_2|u_2) $.
\end{Cor}

\begin{proof}
Since the common message rate can be exchanged with any user's
independent message rate, we set
$R_0=\alpha+\beta,R_1^{\prime}=R_1+\alpha, R_2^{\prime}=R_2+\beta$
where $\alpha,\beta\geq 0$. Plugging these expressions into the
rates in Theorem~\ref{Theorem_product_wiretap_channel} and using
Fourier-Moztkin elimination, we get the region given in the
corollary.
\end{proof}

\begin{Remark}
If we disable the eavesdropper by setting $Z_{11}=Z_{22}=\phi$, we
recover the capacity region of the underlying broadcast channel
without a common message, which was found originally in
\cite{Poltyrev}.
\end{Remark}

At this point, one may ask whether the results of this section can
be extended to arbitrary numbers of users and parallel
sub-channels. Once we have
Theorem~\ref{Theorem_product_wiretap_channel}, the extension of
the results to an arbitrary number of parallel sub-channels is
rather straightforward. Let us consider the parallel degraded
multi-receiver wiretap channel with $M$ sub-channels, and in each
sub-channel, we have either the following Markov chain
\begin{align}
X_l\rightarrow Y_{1l} \rightarrow Y_{2l} \rightarrow Z_l
\label{first_Markov_chain}
\end{align}
or this Markov chain
\begin{align}
X_l\rightarrow Y_{2l} \rightarrow Y_{1l} \rightarrow Z_l
\label{second_Markov_chain}
\end{align}
for any $l\in\{1,\ldots,M\}$. We define the set of indices
$\mathcal{S}_1$ (resp. $\mathcal{S}_2$) as those where for every
$l\in\mathcal{S}_1$ (resp. $l\in\mathcal{S}_2$), the Markov chain
in (\ref{first_Markov_chain}) (resp. in
(\ref{second_Markov_chain})) is satisfied. Then, using
Theorem~\ref{Theorem_product_wiretap_channel}, we obtain the
secrecy capacity region of the channel with two users and $M$
sub-channels as given in the following theorem which is proved in
Appendix~\ref{proof_of_D_PWC_M}.

\begin{Theo}
\label{Theorem_D_PWC_M} The secrecy capacity region of the
parallel degraded multi-receiver wiretap channel with $M$
sub-channels, where each sub-channel satisfies either
(\ref{first_Markov_chain}) or (\ref{second_Markov_chain}) is given
by the union of the rate tuples $(R_0,R_1,R_2)$ satisfying
\begin{align}
R_0&\leq \sum_{l=1}^M I(U_l;Y_{1l}|Z_l)\\
R_0&\leq \sum_{l=1}^M I(U_l;Y_{2l}|Z_l)\\
R_0+R_1 &\leq \sum_{l\in\mathcal{S}_1}I(X_l;Y_{1l}|Z_l)+\sum_{l\in\mathcal{S}_2} I(U_l;Y_{1l}|Z_l)\\
R_0+R_2 &\leq \sum_{l\in\mathcal{S}_2}I(X_l;Y_{2l}|Z_l)+\sum_{l\in\mathcal{S}_1} I(U_l;Y_{2l}|Z_l)\\
R_0+R_1+R_2 &\leq \sum_{l\in\mathcal{S}_1}I(X_l;Y_{1l}|Z_l)+\sum_{l\in\mathcal{S}_2} I(U_l;Y_{1l}|Z_l)+\sum_{l\in\mathcal{S}_2}I(X_l;Y_{2l}|U_l,Z_l)\\
R_0+R_1+R_2 &\leq
\sum_{l\in\mathcal{S}_2}I(X_l;Y_{2l}|Z_l)+\sum_{l\in\mathcal{S}_1}
I(U_l;Y_{2l}|Z_l)+\sum_{l\in\mathcal{S}_1}I(X_l;Y_{1l}|U_l,Z_l)
\end{align}
where the union is over all distributions of the form
$\prod_{l=1}^M p(u_l,x_l)$.
\end{Theo}

We are now left with the question whether these results can be
generalized to an arbitrary number of users. If we consider the
parallel degraded multi-receiver wiretap channel with more than
two sub-channels and an arbitrary number of users, the secrecy
capacity region for the scenario where each user receives a common
message in addition to an independent message does not seem to be
characterizable. Our intuition comes from the fact that, as of
now, the capacity region of the corresponding broadcast channel
without secrecy constraints is unknown \cite{Goldsmith_Effros}.
However, if we consider the scenario where each user receives only
an independent message, i.e., there is no common message, then the
secrecy capacity region may be found, because the capacity region
of the corresponding broadcast channel without secrecy constraints
can be established \cite{Goldsmith_Effros}, although there is no
explicit expression for it in the literature. We expect this
particular generalization to be rather straightforward, and do not
pursue it here.

\section{Sum of Degraded Multi-receiver Wiretap Channels}

We now consider a different multi-receiver wiretap channel which
can be viewed as a sum of two degraded multi-receiver wiretap
channels with two users and one eavesdropper. In this channel
model, the transmitter has two non-intersecting input alphabets,
i.e., $\mathcal{X}_1,\mathcal{X}_2$ with $\mathcal{X}_1\cap
\mathcal{X}_2=\emptyset$, and each receiver has two
non-intersecting alphabets, i.e.,
$\mathcal{Y}_{j1},\mathcal{Y}_{j2}$ with $\mathcal{Y}_{j1}\cap
\mathcal{Y}_{j2}=\emptyset$ for the $j$th user, $j=1,2$, and
$\mathcal{Z}_{1},\mathcal{Z}_{2}$ with $\mathcal{Z}_{1}\cap
\mathcal{Z}_{2}=\emptyset$ for the eavesdropper. The channel is
again memoryless with transition probability
\begin{equation}
p(y_1,y_2,z|x)= \left\{
\begin{array}{ll}
p(y_{11}|x_1)p(y_{21}|y_{11})p(z_1|y_{21})&\textrm{if }
(x,y_1,y_2,z)\in\mathcal{X}_1\times\mathcal{Y}_{11}\times\mathcal{Y}_{21}
\times\mathcal{Z}_1\\
p(y_{22}|x_2)p(y_{12}|y_{22})p(z_2|y_{12})&\textrm{if }
(x,y_1,y_2,z)\in\mathcal{X}_2\times\mathcal{Y}_{21}\times\mathcal{Y}_{22}
\times\mathcal{Z}_2\\
0& \textrm{otherwise}
\end{array}\right.
\end{equation}
where $x\in\mathcal{X}=\mathcal{X}_1\cup\mathcal{X}_2$,
$y_j\in\mathcal{Y}_j=\mathcal{Y}_{j1}\cup\mathcal{Y}_{j2}, j=1,2$
and $z\in\mathcal{Z}=\mathcal{Z}_1\cup\mathcal{Z}_2$. Thus, if the
transmitter chooses to use its first alphabet, i.e.,
$\mathcal{X}_1$, the second user (resp. eavesdropper) receives a
degraded version of user 1's (resp. user 2's) observation.
However, if the transmitter uses its second alphabet, i.e.,
$\mathcal{X}_2$, the first user (resp. eavesdropper) receives a
degraded version of user 2's (resp. user 1's) observation.
Consequently, the overall channel is not degraded from any user's
perspective, however it is degraded from the eavesdropper's
perspective.

A $(2^{nR_0},2^{nR_1},2^{nR_2},n)$ code for this channel consists
of three message sets,
$w_0\in\mathcal{W}_0=\{1,\ldots,2^{nR_0}\}$,
$w_j\in\mathcal{W}_j=\{1,\ldots,2^{nR_j}\}, j=1,2$, one encoder
$f:\mathcal{W}_0\times \mathcal{W}_1\times\mathcal{W}_2\rightarrow
\mathcal{X}^n$ and two decoders, one at each legitimate receiver,
$g_j:\mathcal{Y}_j^n\rightarrow \mathcal{W}_0\times\mathcal{W}_j,
j=1,2$. The probability of error is defined as
$P_e^n=\max_{j=1,2}\Pr\left[g_j(Y_j^n)\neq (W_0,W_j)\right]$. A
rate tuple $(R_0,R_1,R_2)$ is said to be achievable if there
exists a code with $\lim_{n\rightarrow\infty}P_e^n=0$ and
\begin{align}
\lim_{n\rightarrow\infty}\frac{1}{n}H(\mathcal{S}(W)|Z^n)\geq
\sum_{j\in\mathcal{S}(W)}R_j,\quad \forall ~\mathcal{S}(W)
\end{align}
where $\mathcal{S}(W)$ denotes any subset of $\{W_0,W_1,W_2\}$.
The secrecy capacity region is the closure of all achievable
secrecy rate tuples.

The secrecy capacity region of this channel is given in the
following theorem which is proved in
Appendix~\ref{Proof_Sum_Unmatched}.

\begin{Theo}
\label{Theorem_Sum_Unmatched} The secrecy capacity region of the
sum of two degraded multi-receiver wiretap channels is given by
the union of the rate tuples $(R_0,R_1,R_2)$ satisfying
\begin{align}
R_0&\leq \alpha I(U_1;Y_{11}|Z_1)+\bar{\alpha}I(U_2;Y_{12}|Z_2)\\
R_0&\leq \alpha I(U_1;Y_{21}|Z_1)+\bar{\alpha}I(U_2;Y_{22}|Z_2)\\
R_0+R_1&\leq \alpha
I(X_1;Y_{11}|Z_1)+\bar{\alpha}I(U_2;Y_{12}|Z_2)\\
R_0+R_2&\leq \alpha
I(U_1;Y_{21}|Z_1)+\bar{\alpha}I(X_2;Y_{22}|Z_2)\\
R_0+R_1+R_2 & \leq \alpha
I(X_1;Y_{11}|Z_1)+\bar{\alpha}I(U_2;Y_{12}|Z_2)+\bar{\alpha}I(X_2;Y_{22}|U_2,Z_2)\\
R_0+R_1+R_2 &\leq \alpha I(U_1;Y_{21}|Z_1)+\alpha
I(X_1;Y_{11}|U_1,Z_1)+\bar{\alpha}I(X_2;Y_{22}|Z_2)
\end{align}
where the union is over all $\alpha\in[0,1]$ and distributions of
the form $ p(u_1,u_2,x_1,x_2)=p(u_1,x_1)p(u_2,x_2) $.
\end{Theo}

\begin{Remark}
This channel model is similar to the parallel degraded
multi-receiver wiretap channel of the previous section in the
sense that it can be viewed to consist of two parallel
sub-channels, however now the transmitter cannot use both
sub-channels simultaneously. Instead, it should invoke a
time-sharing approach between these two so-called parallel
sub-channels ($\alpha$ reflects this concern). Moreover,
superposition coding scheme again achieves the boundary of the
secrecy capacity region, however it differs from the standard one
\cite{cover_book} in the sense that it needs to be modified to
incorporate secrecy constraints, i.e., it needs to use stochastic
encoding to associate each message with multiple codewords.
\end{Remark}

\begin{Remark}
An interesting point about the secrecy capacity region is that if
we drop the secrecy constraints by setting $Z_1=Z_2=\phi$, we are
unable to recover the capacity region of the corresponding
broadcast channel that was found in~\cite{Product_Broadcast}.
After setting $Z_1=Z_2=\phi$, we note that each expression in
Theorem~\ref{Theorem_Sum_Unmatched} and its counterpart describing
the capacity region \cite{Product_Broadcast} differ by exactly
$h(\alpha)$. The reason for this is as follows. Here, $\alpha$ not
only denotes the time-sharing variable but also carries an
additional information, i.e., the change of the channel that is in
use is part of the information transmission. However, since the
eavesdropper can also decode these messages, the term $h(\alpha)$,
which is the amount of information that can be transmitted via
changes of the channel in use, disappears in the secrecy capacity
region.
\end{Remark}

\section{Conclusions}
In this paper, we studied secure broadcasting to many users in the
presence of an eavesdropper. Characterizing the secrecy capacity
region of this channel in its most general form seems to be
intractable for now, since the version of this problem without any
secrecy constraints, is the broadcast channel with an arbitrary
number of receivers, whose capacity region is open. Consequently,
we took the approach of considering special classes of channels.
In particular, we considered degraded multi-receiver wiretap
channels, parallel multi-receiver wiretap channels with a more
noisy eavesdropper, parallel multi-receiver wiretap channels with
less noisiness orderings in each sub-channel, and parallel
degraded multi-receiver wiretap channels. For each channel model,
we obtained either partial characterization of the secrecy
capacity region or the entire region.

\appendices
\appendixpage

\section{Proofs of Theorem~\ref{less_noisy_nultiuser_wiretap} and Corollary~\ref{degraded_multiuser_wiretap}}

\subsection{Proof of Theorem~\ref{less_noisy_nultiuser_wiretap}}
\label{proof_of_less_noisy_multiuser_wiretap}

First, we show achievability, then provide the converse.

\subsubsection{Achievability} Fix the probability distribution as
\begin{align}
p(u_1)p(u_2|u_1)\ldots p(u_{K-1}|u_{K-2})p(x|u_{K-1})
\label{fixed_dist}
\end{align}

\noindent \textbf{\underline{Codebook generation:}}

\begin{itemize}
\item Generate $2^{n(R_0+R_1+\tilde{R}_1)}$ length-$n$ sequences
$\bu_1$ through $p(\bu_1)=\prod_{i=1}^n p(u_{1,i})$ and index them
as $\bu_1(w_0,w_1,\tilde{w}_1)$ where
$w_0\in\left\{1,\ldots,2^{nR_0}\right\}$,
$w_1\in\left\{1,\ldots,2^{nR_1}\right\}$ and
$\tilde{w}_1\in\big\{1,\ldots,2^{n\tilde{R}_1}\big\}$.

\item For each $\bu_{j-1}$, where $j=2,\ldots,K-1$, generate
$2^{n(R_j+\tilde{R}_j)}$ length-$n$ sequences $\bu_j$ through
$p(\bu_j| \bu_{j-1})=\prod_{i=1}^n p(u_{j,i}|u_{j-1,i})$ and index
them as $\bu_j
(w_0,w_1,\ldots,w_j,\tilde{w}_1,\ldots,\tilde{w}_j)$ where
$w_j\in\left\{1,\ldots,2^{nR_j}\right\}$ and
$\tilde{w}_j\in\big\{1,\ldots,2^{n\tilde{R}_j}\big\}$.

\item Finally, for each $\bu_{K-1}$, generate
$2^{n(R_K+\tilde{R}_K)}$ length-$n$ sequences $\bx$ through
$p(\bx|\bu_{K-1})=\prod_{i=1}^n p(x_i|u_{K,i})$ and index them as
$\bx(w_0,w_1,\ldots,w_K,\tilde{w}_1,\ldots,\tilde{w}_K)$ where
$w_K\in\break\left\{1,\ldots,2^{nR_K}\right\}$ and
$\tilde{w}_K\in\big\{1,\ldots,2^{n\tilde{R}_K}\big\}$.

\item Furthermore, we set
\begin{align}
\tilde{R}_i=I(U_i;Z|U_{i-1}),\quad i=1,\ldots,K
\label{dummy_codeword_rate_1}
\end{align}
where $U_0=\phi$ and $U_K=X$.
\end{itemize}

\noindent \textbf{\underline{Encoding:}}

\vspace{0.25cm} Assume the messages to be transmitted are
$\left(w_0,w_1,\ldots,w_K\right)$. Then, the encoder randomly
picks a set $\left(\tilde{w}_1,\ldots,\tilde{w}_K\right)$ and
sends $\bx(w_0,w_1,\ldots,w_K,\tilde{w}_1,\ldots,\tilde{w}_K)$.

\vspace{0.5cm} \noindent \textbf{\underline{Decoding:}}

\vspace{0.25cm} It is straightforward to see that if the following
conditions are satisfied,
\begin{align}
R_0+R_1+\tilde{R}_1 & \leq I(U_1;Y_1) \\
R_j+\tilde{R}_j & \leq I(U_j;Y_j|U_{j-1}), \quad j=2,\ldots,K-1 \\
R_K+\tilde{R}_K & \leq I(X;Y_K|U_{K-1})
\end{align}
then all users can decode both the common message and the
independent message directed to itself with vanishingly small
error probability. Moreover, since the channel is degraded, each
user, say the $j$th one, can decode all of the independent
messages intended for the users whose channels are degraded with
respect to the $j$th user's channel. Thus, these degraded users'
rates can be exploited to increase the $j$th user's rate which
leads to the following achievable region
\begin{align}
R_0+\sum_{j=1}^\ell R_j +\sum_{j=1}^\ell \tilde{R}_j & \leq
\sum_{j=1}^\ell I(U_j;Y_j|U_{j-1}), \quad \ell=1,\ldots,K
\label{proof_degraded_multiuser_ach_region_1}
\end{align}
where $U_0=\phi$ and $U_K=X$. Moreover, after eliminating
$\big\{\tilde{R}_j\big\}_{j=1}^K$ ,
(\ref{proof_degraded_multiuser_ach_region_1}) can be expressed as
\begin{align}
R_0+\sum_{j=1}^\ell R_j & \leq \sum_{j=1}^\ell
I(U_j;Y_j|U_{j-1})-I(U_{\ell};Z), \quad \ell=1,\ldots,K
\end{align}
where we used the fact that
\begin{align}
\sum_{j=1}^{\ell}\tilde{R}_j=\sum_{j=1}^{\ell}I(U_j;Z|U_{j-1})=I(U_1,\ldots,U_{\ell};Z)=I(U_{\ell};Z)
\label{superposition_coding_implies}
\end{align}
where the second and the third equalities are due to the following
Markov chain
\begin{align}
U_1\rightarrow \ldots\rightarrow U_{K-1}\rightarrow X \rightarrow
Z
\end{align}

\vspace{0.5cm} \noindent \textbf{\underline{Equivocation
calculation:}}

\vspace{0.25cm} We now calculate the equivocation of the code
described above. To that end, we first introduce the following
lemma which states that a code satisfying the sum rate secrecy
constraint fulfills all other secrecy constraints.

\begin{Lem}
\label{lemma_sum_rate_sufficient} If the sum rate secrecy
constraint is satisfied, i.e.,
\begin{align}
\frac{1}{n}H(W_0,W_1,\ldots,W_K|Z^n)\geq \sum_{j=0}^K
R_j-\epsilon_n \label{lemma_sum_rate_satisfied}
\end{align}
then all other secrecy constraints are satisfied as well, i.e.,
\begin{align}
\frac{1}{n}H(\mathcal{S}(W)|Z^n)\geq \sum_{ j\in\mathcal{S}(W)}
R_j-\epsilon_n
\end{align}
where $\mathcal{S}(W)$ denotes any subset of
$\left\{W_0,W_1,\ldots,W_K\right\}$.
\end{Lem}
\begin{proof}
The proof of this lemma is as follows.
\begin{align}
\frac{1}{n}H(\mathcal{S}(W)|Z^n)&=\frac{1}{n}H(\mathcal{S}(W),\mathcal{S}^c
(W)|Z^n)-\frac{1}{n}H(\mathcal{S}^c (W)|\mathcal{S}(W),Z^n)\\
&\geq \sum_{j=0}^K R_j-\epsilon_n -\frac{1}{n}H(\mathcal{S}^c
(W)|\mathcal{S}(W),Z^n) \label{lemma_sum_rate_satisfied_assumption} \\
& = \sum_{j\in\mathcal{S}(W)}
R_j-\epsilon_n+\sum_{j\in\mathcal{S}^c(W)} R_j
-\frac{1}{n}H(\mathcal{S}^c
(W)|\mathcal{S}(W),Z^n) \\
& = \sum_{j\in\mathcal{S}(W)} R_j-\epsilon_n+\frac{1}{n}
H(\mathcal{S}^c(W)) -\frac{1}{n}H(\mathcal{S}^c
(W)|\mathcal{S}(W),Z^n) \label{iid_uniform_message_sets} \\
& \geq \sum_{j\in\mathcal{S}(W)} R_j-\epsilon_n
\end{align}
where (\ref{lemma_sum_rate_satisfied_assumption}) is due to the
fact that we assumed that sum rate secrecy constraint
(\ref{lemma_sum_rate_satisfied}) is satisfied and
(\ref{iid_uniform_message_sets}) follows from
\begin{align}
\sum_{j\in\mathcal{S}^c(W)} R_j & = \frac{1}{n}
H(\mathcal{S}^c(W))
\end{align}
which is a consequence of the fact that message sets are uniformly
and independently distributed.
\end{proof}

Hence, it is sufficient to check whether coding scheme presented
satisfies the sum rate secrecy constraint.
\begin{align}
\lefteqn{H(W_0,W_1,\ldots,W_K|Z^n)= H(W_0,W_1,\ldots,W_K,Z^n)-H(Z^n)}&\\
&= H(U_1^n,\ldots,U_{K-1}^n,X^n,W_0,W_1,\ldots,W_K,Z^n)-H(Z^n)
\nonumber \\
&\quad -H(U_1^n,\ldots,U_{K-1}^n,X^n|W_0,W_1,\ldots,W_K,Z^n)\\
&=
H(U_1^n,\ldots,U_{K-1}^n,X^n)+H(W_0,W_1,\ldots,W_K,Z^n|U_1^n,\ldots,U_{K-1}^n,X^n)-H(Z^n)
\nonumber \\
&\quad -H(U_1^n,\ldots,U_{K-1}^n,X^n|W_0,W_1,\ldots,W_K,Z^n)\\
&\geq
H(U_1^n,\ldots,U_{K-1}^n,X^n)-I(U_1^n,\ldots,U_{K-1}^n,X^n;Z^n)\nonumber\\
&\quad -H(U_1^n,\ldots,U_{K-1}^n,X^n|W_0,W_1,\ldots,W_K,Z^n)
\label{proof_degraded_equi_calc_step_1}
\end{align}
where each term will be treated separately. Since given
$U_k^n=u_k^n$, $U_{k+1}^n$ can take
$2^{n(R_{k+1}+\tilde{R}_{k+1})}$ values uniformly, the first term
is
\begin{align}
H(U_1^n,\ldots,U_{K-1}^n,X^n)&=H(U_1^n)+\sum_{k=2}^{K-1}H(U_{k}^n|U_{k-1}^n)+H(X^n|U_{K-1}^n)\\
&=nR_0+n\sum_{k=1}^K R_k + n\sum_{k=1}^K \tilde{R}_k
\label{proof_degraded_equi_calc_step_2}
\end{align}
where the first equality follows from the following Markov chain
\begin{align}
U_1^n\rightarrow U_2^n\rightarrow\ldots\rightarrow
U_{K-1}^n\rightarrow X^n
\label{proof_degraded_mu_wt_codebook_chain}
\end{align}
The second term in (\ref{proof_degraded_equi_calc_step_1}) is
\begin{align}
I(U_1^n,\ldots,U_{K-1}^n,X^n;Z^n)&=I(X^n;Z^n)+I(U_1^n,U_2^n,\ldots,U_{K-1}^n;Z^n|X^n)
\\
&= I(X^n;Z^n) \label{proof_degraded_mu_wt_codebook_chain_implies}\\
&\leq nI(X;Z)+\gamma_n\label{memoryless_channel_2}
\end{align}
where (\ref{proof_degraded_mu_wt_codebook_chain_implies}) follows
from the Markov chain in
(\ref{proof_degraded_mu_wt_codebook_chain}) and
(\ref{memoryless_channel_2}) can be shown by following the
approach devised in \cite{Wyner}. We now bound the third term in
(\ref{proof_degraded_equi_calc_step_1}). To that end, assume that
the eavesdropper tries to decode
$\left(U_1^n,\ldots,U_{K-1}^n,X^n\right)$ using the side
information $\left(W_0,W_1,\ldots,W_K\right)$ which is equivalent
to decoding $\left(\tilde{W}_1,\ldots,\tilde{W}_K\right)$. Since
$\tilde{R}_j$s are selected to ensure that the eavesdropper can
decode them successively, see (\ref{dummy_codeword_rate_1}), then
using Fano's lemma, we have
\begin{align}
H(U_1^n,\ldots,U_{K-1}^n,X^n|W_0,W_1,\ldots,W_K,Z^n) \leq
\epsilon_n \label{proof_degraded_equi_calc_step_3}
\end{align}
Thus, using (\ref{proof_degraded_equi_calc_step_2}),
(\ref{memoryless_channel_2}) and
(\ref{proof_degraded_equi_calc_step_3}) in
(\ref{proof_degraded_equi_calc_step_1}), we get
\begin{align}
H(W_0,W_1,\ldots,W_K|Z^n)&\geq n\sum_{j=0}^K R_j +n\sum_{j=1}^K
\tilde{R}_j-nI(X;Z)-\epsilon_n \\
&= n\sum_{j=0}^K R_j -\epsilon_n-\gamma_n
\label{proof_degraded_equi_calc_step_4}
\end{align}
where (\ref{proof_degraded_equi_calc_step_4}) follows from the
following, see (\ref{dummy_codeword_rate_1}) and
(\ref{superposition_coding_implies}),
\begin{align}
\sum_{j=1}^K\tilde{R}_j= I(X;Z)
\end{align}

\subsubsection{Converse}

First let us define the following auxiliary random variables,
\begin{align}
U_{k,i}=W_0W_1\ldots W_k Y_{k+1}^{i-1}Z_{i+1}^n,\quad
k=1,\ldots,K-1
\end{align}
which satisfy the following Markov chain
\begin{align}
U_{1,i}\rightarrow U_{2,i}\rightarrow \ldots\rightarrow
U_{K-1,i}\rightarrow X_i \rightarrow
\left(Z_i,Y_{K,i},\ldots,Y_{1,i}\right)
\end{align}
To provide a converse, we will show
\begin{align}
\frac{1}{n}H(W_0,W_1,\ldots,W_{\ell}|Z^n) \leq
\sum_{k=1}^{\ell}I(U_k;Y_k|U_{k-1})-I(U_{\ell};Z), \quad
\ell=1,\ldots,K \label{converse_dummy}
\end{align}
where $U_0=\phi$, $U_K=X$. We show this in three steps. First, let
us write down
\begin{align}
H(W_0,W_1,\ldots,W_{\ell}|Z^n)=H(W_0,W_1|Z^n)+\sum_{k=2}^{\ell}H(W_k|W_0,W_1,\ldots,W_{k-1},Z^n)
\label{dummy_decompose_less_noisy_1}
\end{align}
The first term on the right hand side of
(\ref{dummy_decompose_less_noisy_1}) is bounded as follows,
\begin{align}
\lefteqn{H(W_0,W_1|Z^n)\leq I(W_0,W_1;Y_1^n)-I(W_0,W_1;Z^n)+\epsilon_n}  \label{Fano_10} \\
&\leq
\sum_{i=1}^{n}I(W_0,W_1;Y_{1,i}|Y_1^{i-1},Z_{i+1}^n)-I(W_0,W_1;Z_i|Y_1^{i-1},Z_{i+1}^n)
+\epsilon_n \label{Fano_and_Csiszar_identity}\\
& \leq
\sum_{i=1}^{n}I(W_0,W_1;Y_{1,i}|Y_1^{i-1},Z_{i+1}^n)-I(W_0,W_1;Z_i|Y_1^{i-1},Z_{i+1}^n)\nonumber\\
&\quad + I(Y_1^{i-1},Z_{i+1}^n;Y_{1,i})-I(Y_1^{i-1},Z_{i+1}^n;Z_i)
+\epsilon_n \label{less_noisy_nultiuser_wiretap_less_noisy_1}\\
&=\sum_{i=1}^{n}I(W_0,W_1,Y_1^{i-1},Z_{i+1}^n;Y_{1,i})-I(W_0,W_1,Y_1^{i-1},Z_{i+1}^n;Z_i)
+\epsilon_n  \\
&\leq \sum_{i=1}^{n}I(W_0,W_1,Y_1^{i-1},Z_{i+1}^n;Y_{1,i})-I(W_0,W_1,Y_1^{i-1},Z_{i+1}^n;Z_i)\nonumber\\
&\quad
+I(Y_2^{i-1};Y_{1,i}|W_0,W_1,Y_1^{i-1},Z_{i+1}^n)-I(Y_2^{i-1};Z_i|W_0,W_1,Y_1^{i-1},Z_{i+1}^n)
+\epsilon_n \label{less_noisy_nultiuser_wiretap_less_noisy_2} \\
& = \sum_{i=1}^{n}I(W_0,W_1,Y_1^{i-1},Z_{i+1}^n,Y_2^{i-1};Y_{1,i})
-I(W_0,W_1,Y_1^{i-1},Z_{i+1}^n,Y_2^{i-1};Z_i)
+\epsilon_n \\
& =
\sum_{i=1}^{n}I(W_0,W_1,Z_{i+1}^n,Y_2^{i-1};Y_{1,i})-I(W_0,W_1,Z_{i+1}^n,Y_2^{i-1};Z_i)\\
&\quad
+I(Y_1^{i-1};Y_{1,i}|W_0,W_1,Z_{i+1}^n,Y_2^{i-1})-I(Y_1^{i-1};Z_i|W_0,W_1,Z_{i+1}^n,Y_2^{i-1})
+ \epsilon_n \\
& =
\sum_{i=1}^{n}I(W_0,W_1,Z_{i+1}^n,Y_2^{i-1};Y_{1,i})-I(W_0,W_1,Z_{i+1}^n,Y_2^{i-1};Z_i)+
\epsilon_n
\label{proof_less_noisy_channel_degraded}\\
& = \sum_{i=1}^{n}I(U_{1,i};Y_{1,i})-I(U_{1,i};Z_i)+ \epsilon_n
\label{dummy_decompose_less_noisy_2}
\end{align}
where (\ref{Fano_10}) follows from Fano's lemma,
(\ref{Fano_and_Csiszar_identity}) is obtained using Csiszar-Korner
identity (see Lemma~7 of \cite{Korner}),
(\ref{less_noisy_nultiuser_wiretap_less_noisy_1}) is due to the
fact that
\begin{align}
I(Y_1^{i-1},Z_{i+1}^n;Y_{1,i})-I(Y_1^{i-1},Z_{i+1}^n;Z_i)>0
\end{align}
which follows from the fact that each user's channel is less noisy
with respect to the eavesdropper. Similarly,
(\ref{less_noisy_nultiuser_wiretap_less_noisy_2}) follows from the
fact that
\begin{align}
I(Y_2^{i-1};Y_{1,i}|W_0,W_1,Y_1^{i-1},Z_{i+1}^n)-I(Y_2^{i-1};Z_i|W_0,W_1,Y_1^{i-1},Z_{i+1}^n)>0
\end{align}
which is a consequence of the fact that each user's channel is
less noisy with respect to the eavesdropper's channel. Finally,
(\ref{proof_less_noisy_channel_degraded}) is due to the following
Markov chain
\begin{align}
Y_1^{i-1}\rightarrow Y_2^{i-1}\rightarrow \left(W_0,W_1,Z_{i+1}^n,
Y_{1,i},Z_i\right)
\end{align}
which is a consequence of the fact that the legitimate receivers
exhibit a degradation order.

We now bound the terms of the summation in
(\ref{dummy_decompose_less_noisy_1}) for $2\leq k \leq K-1$. Let
us use the shorthand notation,
$\tilde{W}_{k-1}=(W_0,W_1,\ldots,W_{k-1})$, then
\begin{align}
\lefteqn{H(W_k|\tilde{W}_{k-1},Z^n)\leq I(W_k;Y_k^n|\tilde{W}_{k-1})-I(W_k;Z^n|\tilde{W}_{k-1})+\epsilon_n} &\label{Fano_11}\\
& \leq\sum_{i=1}^n
I(W_k;Y_{k,i}|\tilde{W}_{k-1},Y_k^{i-1},Z_{i+1}^n)
-I(W_k;Z_{i}|\tilde{W}_{k-1},Y_k^{i-1},Z_{i+1}^n)
+\epsilon_n \label{Fano_and_Csiszar_identity_1}\\
&\leq\sum_{i=1}^n
I(W_k;Y_{k,i}|\tilde{W}_{k-1},Y_k^{i-1},Z_{i+1}^n)
-I(W_k;Z_{i}|\tilde{W}_{k-1},Y_k^{i-1},Z_{i+1}^n)\nonumber\\
&\quad +
I(Y_{k+1}^{i-1};Y_{k,i}|\tilde{W}_{k-1},Y_k^{i-1},Z_{i+1}^n,W_k)
-I(Y_{k+1}^{i-1};Z_{i}|\tilde{W}_{k-1},Y_k^{i-1},Z_{i+1}^n,W_k)+\epsilon_n \label{less_noisy_nultiuser_wiretap_less_noisy_3} \\
&=\sum_{i=1}^n
I(W_k,Y_{k+1}^{i-1};Y_{k,i}|\tilde{W}_{k-1},Y_k^{i-1},Z_{i+1}^n)
-I(W_k,Y_{k+1}^{i-1};Z_{i}|\tilde{W}_{k-1},Y_k^{i-1},Z_{i+1}^n)
+\epsilon_n \\
&=\sum_{i=1}^n
I(U_{k,i};Y_{k,i}|U_{k-1,i})-I(U_{k,i};Z_{i}|U_{k-1,i})+\epsilon_n
\label{dummy_decompose_less_noisy_3}
\end{align}
where (\ref{Fano_11}) follows from Fano's lemma,
(\ref{Fano_and_Csiszar_identity_1}) is obtained through
Csiszar-Korner identity, and
(\ref{less_noisy_nultiuser_wiretap_less_noisy_3}) is a consequence
of the fact that
\begin{align}
I(Y_{k+1}^{i-1};Y_{k,i}|\tilde{W}_{k-1},Y_k^{i-1},Z_{i+1}^n,W_k)
-I(Y_{k+1}^{i-1};Z_{i}|\tilde{W}_{k-1},Y_k^{i-1},Z_{i+1}^n,W_k)>0
\end{align}
which follows from the fact that each user's channel is less noisy
with respect to the eavesdropper's channel. Finally, we bound the
following term where we again use the shorthand notation
$\tilde{W}_{K-1}=(W_0,W_1,\ldots,W_{K-1})$,

\begin{align}
\lefteqn{H(W_K|\tilde{W}_{K-1},Z^n)\leq I(W_K;Y_K^n|\tilde{W}_{K-1})-I(W_K;Z^n|\tilde{W}_{K-1})+\epsilon_n}&\label{Fano_12} \\
&\leq \sum_{i=1}^n
I(W_K;Y_{K,i}|\tilde{W}_{K-1},Y_K^{i-1},Z_{i+1}^n)-
I(W_K;Z_i|\tilde{W}_{K-1},Y_K^{i-1},Z_{i+1}^n)
+\epsilon_n \label{Fano_and_Csiszar_identity_2} \\
&\leq \sum_{i=1}^n
I(W_K;Y_{K,i}|\tilde{W}_{K-1},Y_K^{i-1},Z_{i+1}^n)
- I(W_K;Z_i|\tilde{W}_{K-1},Y_K^{i-1},Z_{i+1}^n)\nonumber \\
&\quad +I(X_i;Y_{K,i}|\tilde{W}_{K-1},Y_K^{i-1},Z_{i+1}^n,W_K) -
I(X_i;Z_i|\tilde{W}_{K-1},Y_K^{i-1},Z_{i+1}^n,W_K)
+\epsilon_n \label{less_noisy_nultiuser_wiretap_less_noisy_4}\\
&= \sum_{i=1}^n
I(W_K,X_i;Y_{K,i}|\tilde{W}_{K-1},Y_K^{i-1},Z_{i+1}^n) -
I(W_K,X_i;Z_i|\tilde{W}_{K-1},Y_K^{i-1},Z_{i+1}^n)+\epsilon_n
\\
&= \sum_{i=1}^n I(X_i;Y_{K,i}|\tilde{W}_{K-1},Y_K^{i-1},Z_{i+1}^n)
+I(W_K;Y_{K,i}|\tilde{W}_{K-1},Y_K^{i-1},Z_{i+1}^n,X_i)\nonumber
\\
&\quad - I(X_i;Z_i|\tilde{W}_{K-1},Y_K^{i-1},Z_{i+1}^n)
 - I(W_K;Z_i|\tilde{W}_{K-1},Y_K^{i-1},Z_{i+1}^n,X_i)
+\epsilon_n \\
&= \sum_{i=1}^n I(X_i;Y_{K,i}|\tilde{W}_{K-1},Y_K^{i-1},Z_{i+1}^n)
- I(X_i;Z_i|\tilde{W}_{K-1},Y_K^{i-1},Z_{i+1}^n)
+\epsilon_n \label{proof_less_noisy_channel_degraded_1}\\
&= \sum_{i=1}^n I(X_i;Y_{K,i}|U_{K-1,i})- I(X_i;Z_i|U_{K-1,i})
+\epsilon_n \label{dummy_decompose_less_noisy_4}
\end{align}
where (\ref{Fano_12}) follows from Fano's lemma,
(\ref{Fano_and_Csiszar_identity_2}) is obtained by using
Csiszar-Korner identity, and
(\ref{less_noisy_nultiuser_wiretap_less_noisy_4}) follows from the
fact that
\begin{align}
I(X_i;Y_{K,i}|\tilde{W}_{K-1},Y_K^{i-1},Z_{i+1}^n,W_K) -
I(X_i;Z_i|\tilde{W}_{K-1},Y_K^{i-1},Z_{i+1}^n,W_K)>0
\end{align}
which is due to the fact that each user's channel is less noisy
with respect to the eavesdropper and
(\ref{proof_less_noisy_channel_degraded_1}) is due to the Markov
chain
\begin{align}
\left(Y_{K,i},Z_i\right)\rightarrow X_i \rightarrow
\left(W_0,W_1,\ldots,W_K,Y_K^{i-1},Z_{i+1}^n\right)
\end{align}
which follows from the fact that the channel is memoryless.
Finally, plugging (\ref{dummy_decompose_less_noisy_2}),
(\ref{dummy_decompose_less_noisy_3}) and
(\ref{dummy_decompose_less_noisy_4}) into
(\ref{dummy_decompose_less_noisy_1}), we get
\begin{align}
H(W_0,W_1,\ldots,W_{\ell}|Z^n)\leq n\sum_{k=1}^{\ell}
I(U_k;Y_k|U_{k-1})-nI(U_{\ell};Z),\quad \ell=1,\ldots,K
\end{align}
where $U_0=\phi$ and $U_K=X$, and this concludes the converse.

\subsection{Proof of Corollary~\ref{degraded_multiuser_wiretap}}
\label{proof_of_corollary_degraded_multiuser_wiretap}

First, we note that
\begin{align}
I(U_{\ell};Z)=I(U_1,\ldots,U_{\ell};Z)=\sum_{k=1}^{\ell}I(U_k;Z|U_{k-1})
\label{cor_degraded_1}
\end{align}
where the first equality is due to the following Markov chain
\begin{align}
U_1\rightarrow \ldots \rightarrow U_{K-1}\rightarrow X \rightarrow
Z
\end{align}
By plugging (\ref{cor_degraded_1}) into
(\ref{capacity_less_noisy}), we get
\begin{align}
R_0+\sum_{k=1}^\ell R_k & \leq  \sum_{k=1}^\ell
I(U_k;Y_k|U_{k-1})-I(U_{\ell};Z)\\
&=\sum_{k=1}^\ell I(U_k;Y_k|U_{k-1})-I(U_k;Z|U_{k-1}) \\
&=\sum_{k=1}^\ell I(U_k;Y_k,Z|U_{k-1})-I(U_k;Z|U_{k-1})\label{channel_is_degraded}\\
&=\sum_{k=1}^\ell I(U_k;Y_k|U_{k-1},Z)
\end{align}
where (\ref{channel_is_degraded}) follows from the fact that the
channel is degraded, i.e., we have the following Markov chain
\begin{align}
U_{k-1}\rightarrow U_k \rightarrow Y_k\rightarrow Z
\end{align}

\section{Proofs of Theorem~\ref{common_less_noisier_parallel} and Corollary~\ref{corollary_degraded_common}}

\subsection{Proof of Theorem~\ref{common_less_noisier_parallel}}
\label{proof_of_common_less_noisier_parallel}

Achievability of these rates follows from Proposition~2 of
\cite{Broadcasting_Wornell}. We provide the converse. First let us
define the following random variables,
\begin{align}
Z^n&=\left( Z_1^n,\ldots,Z_M^n\right) \\
Y_k^n&=\left( Y_{k1}^n,\ldots,Y_{kM}^n\right)\\
Z_{i+1}^n&=\left( Z_{1,i+1}^n,\ldots, Z_{M,i+1}^n \right) \\
Y_k^{i-1}&=\left(Y_{k1}^{i-1}\ldots,Y_{kM}^{i-1}\right)\\
Y_k(i)&=\left(Y_{k1}(i),\ldots,Y_{kM}(i)\right)\\
Z(i)&=\left(Z_{1}(i),\ldots,Z_{M}(i)\right)
\end{align}
where $Y_{kl}^{i-1}=(Y_{kl}(1),\ldots,Y_{kl}(i-1))$,
$Z_{l,i+1}^n=\left(Z_l(i+1),\ldots,Z_l(n)\right)$. Start with the
definition,
\begin{align}
H(W_0|Z^n)&=H(W_0)-I(W_0;Z^n)\\
&\leq I(W_0;Y_k^n)-I(W_0;Z^n)+\epsilon_n \label{Fano_2} \\
&= \sum_{i=1}^n
I(W_0;Y_k(i)|Y_k^{i-1})-I(W_0;Z(i)|Z_{i+1}^n)+\epsilon_n \\
&=\sum_{i=1}^n
I(W_0,Z_{i+1}^n;Y_k(i)|Y_k^{i-1})-I(Z_{i+1}^n;Y_k(i)|Y_k^{i-1},W_0)\nonumber\\
&\quad -I(W_0,Y_k^{i-1};Z(i)|Z_{i+1}^n)+I(Y_k^{i-1};Z(i)|Z_{i+1}^n,W_0)+\epsilon_n \\
&=\sum_{i=1}^n
I(W_0,Z_{i+1}^n;Y_k(i)|Y_k^{i-1}) -I(W_0,Y_k^{i-1};Z(i)|Z_{i+1}^n)+\epsilon_n \label{Csiszar_identity_1}\\
&=\sum_{i=1}^n
I(W_0;Y_k(i)|Y_k^{i-1},Z_{i+1}^n)+I(Z_{i+1}^n;Y_k(i)|Y_k^{i-1})\nonumber
\\
&\quad -I(W_0;Z(i)|Z_{i+1}^n,Y_k^{i-1})-I(Y_k^{i-1};Z(i)|Z_{i+1}^n)+\epsilon_n \\
&=\sum_{i=1}^n I(W_0;Y_k(i)|Y_k^{i-1},Z_{i+1}^n)
-I(W_0;Z(i)|Z_{i+1}^n,Y_k^{i-1})+\epsilon_n
\label{Csiszar_identity_2}
\end{align}
where (\ref{Csiszar_identity_1}) and (\ref{Csiszar_identity_2})
are due the following identities
\begin{align}
\sum_{i=1}^n I(Z_{i+1}^n;Y_k(i)|Y_k^{i-1},W_0)&=
\sum_{i=1}^n I(Y_k^{i-1};Z(i)|Z_{i+1}^n,W_0) \\
\sum_{i=1}^n I(Z_{i+1}^n;Y_k(i)|Y_k^{i-1})&= \sum_{i=1}^n
I(Y_k^{i-1};Z(i)|Z_{i+1}^n)
\end{align}
respectively, which are due to Lemma~7 of \cite{Korner}. Now, we
will bound each summand in (\ref{Csiszar_identity_2}) separately.
First, define the following variables.
\begin{align}
U_{k,i}&=\left(Z_{i+1}^n,Y_k^{i-1}\right)\\
\tilde{Y}_k^{l-1}(i)&=\left(Y_{k1}(i),\ldots,Y_{k(l-1)}(i)\right)\\
\tilde{Z}_{l+1}^M (i) & = \left(Z_{l+1}(i),\ldots,Z_M (i)\right)
\end{align}
Hence, the summand in (\ref{Csiszar_identity_2}) can be written as
follows,
\begin{align}
\lefteqn{I(W_0;Y_k(i)|Y_k^{i-1},Z_{i+1}^n)
-I(W_0;Z(i)|Z_{i+1}^n,Y_k^{i-1})}\\
&= I(W_0;Y_k(i)|U_{k,i}) -I(W_0;Z(i)|U_{k,i})\\
&=I(W_0;Y_{k1}(i),\ldots,Y_{kM}(i)|U_{k,i})
-I(W_0;Z_1(i),\ldots,Z_M(i)|U_{k,i})\\
&=\sum_{l=1}^M I(W_0;Y_{kl}(i)|U_{k,i},\tilde{Y}_k^{l-1}(i))
-I(W_0;Z_l(i)|U_{k,i},\tilde{Z}_{l+1}^M (i))\\
&=\sum_{l=1}^M I(W_0,\tilde{Z}_{l+1}^M
(i);Y_{kl}(i)|U_{k,i},\tilde{Y}_k^{l-1}(i))-I(\tilde{Z}_{l+1}^M
(i);Y_{kl}(i)|U_{k,i},\tilde{Y}_k^{l-1}(i),W_0)\nonumber \\
&\quad -I(W_0,\tilde{Y}_k^{l-1}(i);Z_l(i)|U_{k,i},\tilde{Z}_{l+1}^M (i))+I(\tilde{Y}_k^{l-1}(i);Z_l(i)|U_{k,i},\tilde{Z}_{l+1}^M (i),W_0)\\
&=\sum_{l=1}^M I(W_0,\tilde{Z}_{l+1}^M
(i);Y_{kl}(i)|U_{k,i},\tilde{Y}_k^{l-1}(i))-I(W_0,\tilde{Y}_k^{l-1}(i);Z_l(i)|U_{k,i},\tilde{Z}_{l+1}^M
(i)) \label{Csiszar_identity_3} \\
&=\sum_{l=1}^M I(\tilde{Z}_{l+1}^M
(i);Y_{kl}(i)|U_{k,i},\tilde{Y}_k^{l-1}(i))+I(W_0
;Y_{kl}(i)|U_{k,i},\tilde{Y}_k^{l-1}(i),\tilde{Z}_{l+1}^M(i))\nonumber\\
&\quad -I(\tilde{Y}_k^{l-1}(i);Z_l(i)|U_{k,i},\tilde{Z}_{l+1}^M
(i)) -I(W_0;Z_l(i)|U_{k,i},\tilde{Z}_{l+1}^M
(i),\tilde{Y}_k^{l-1}(i))\\
&=\sum_{l=1}^M I(W_0
;Y_{kl}(i)|U_{k,i},\tilde{Y}_k^{l-1}(i),\tilde{Z}_{l+1}^M(i))
-I(W_0;Z_l(i)|U_{k,i},\tilde{Z}_{l+1}^M (i),\tilde{Y}_k^{l-1}(i))
\label{Csiszar_identity_4}
\end{align}
where (\ref{Csiszar_identity_3}) and (\ref{Csiszar_identity_4})
follow from the following identities
\begin{align}
\sum_{l=1}^M I(\tilde{Z}_{l+1}^M
(i);Y_{kl}(i)|U_{k,i},\tilde{Y}_k^{l-1}(i),W_0)&=
\sum_{l=1}^M I(\tilde{Y}_k^{l-1}(i);Z_l(i)|U_{k,i},\tilde{Z}_{l+1}^M (i),W_0)\\
\sum_{l=1}^M I(\tilde{Z}_{l+1}^M
(i);Y_{kl}(i)|U_{k,i},\tilde{Y}_k^{l-1}(i)) &=\sum_{l=1}^M
I(\tilde{Y}_k^{l-1}(i);Z_l(i)|U_{k,i},\tilde{Z}_{l+1}^M (i))
\end{align}
respectively, which are again due to Lemma~7 of \cite{Korner}.
Now, define the set of sub-channels, say $\mathcal{S}(k)$, in
which the $k$th user is less noisy with respect to the
eavesdropper. Thus, the summands in (\ref{Csiszar_identity_4}) for
$l\notin\mathcal{S}(k)$ are negative and by dropping them, we can
bound (\ref{Csiszar_identity_4}) as follows,
\begin{align}
\lefteqn{I(W_0;Y_k(i)|Y_k^{i-1},Z_{i+1}^n)
-I(W_0;Z(i)|Z_{i+1}^n,Y_k^{i-1})} \nonumber
\\
&\leq
\sum_{l\in\mathcal{S}(k)} I(W_0
;Y_{kl}(i)|U_{k,i},\tilde{Y}_k^{l-1}(i),\tilde{Z}_{l+1}^M(i))
-I(W_0;Z_l(i)|U_{k,i},\tilde{Z}_{l+1}^M (i),\tilde{Y}_k^{l-1}(i))
\label{drop_negatives_1}
\end{align}
Moreover, for $l\in\mathcal{S}(k)$, we have
\begin{align}
I(U_{k,i},\tilde{Y}_k^{l-1}(i),\tilde{Z}_{l+1}^M(i);Y_{kl}(i))
-I(U_{k,i},\tilde{Y}_k^{l-1}(i),\tilde{Z}_{l+1}^M(i);Z_{l}(i))&\geq
0 \label{add_positive_1}\\
I(X_l(i)
;Y_{kl}(i)|U_{k,i},\tilde{Y}_k^{l-1}(i),\tilde{Z}_{l+1}^M(i),W_0)
-I(X_l(i);Z_l(i)|U_{k,i},\tilde{Z}_{l+1}^M
(i),\tilde{Y}_k^{l-1}(i),W_0)&\geq 0 \label{add_positive_2}
\end{align}
where both are due to the fact that for $l\in\mathcal{S}(k)$, in
this sub-channel the $k$th user is less noisy with respect to the
eavesdropper. Therefore, adding (\ref{add_positive_1}) and
(\ref{add_positive_2}) to each summand in
(\ref{drop_negatives_1}), we get the following bound,
\begin{align}
\lefteqn{I(W_0;Y_k(i)|Y_k^{i-1},Z_{i+1}^n)
-I(W_0;Z(i)|Z_{i+1}^n,Y_k^{i-1})}\nonumber\\
&\leq \sum_{l\in\mathcal{S}(k)}
I(X_l(i),W_0,U_{k,i},\tilde{Y}_k^{l-1}(i),\tilde{Z}_{l+1}^M(i)
;Y_{kl}(i)) -
I(X_l(i),W_0,U_{k,i},\tilde{Y}_k^{l-1}(i),\tilde{Z}_{l+1}^M(i)
;Z_{l}(i)) \\
&= \sum_{l\in\mathcal{S}(k)} I(X_l(i) ;Y_{kl}(i)) - I(X_l(i)
;Z_{l}(i)) \label{Almost_single_letter}
\end{align}
where the equality follows from the following Markov chain
\begin{align}
\left(W_0,U_{k,i},\tilde{Y}_k^{l-1}(i),\tilde{Z}_{l+1}^M(i)
\right) \rightarrow X_l(i) \rightarrow
\left(Y_{kl}(i),Z_l(i)\right)
\end{align}
which is a consequence of the facts that channel is memoryless and
sub-channels are independent. Finally, using
(\ref{Almost_single_letter}) in (\ref{Csiszar_identity_2}), we get
\begin{align}
H(W_0|Z^n)&\leq \sum_{i=1}^n \sum_{l\in\mathcal{S}(k)} I(X_l(i)
;Y_{kl}(i)) - I(X_l(i) ;Z_{l}(i))+\epsilon_n \\
&\leq n \sum_{l\in\mathcal{S}(k)}  I(X_l ;Y_{kl}) - I(X_l
;Z_{l})+\epsilon_n \\
&= n \sum_{l=1}^M \left[ I(X_l ;Y_{kl}) - I(X_l ;Z_{l})\right]^+
+\epsilon_n
\end{align}
which completes the proof.

\subsection{Proof of Corollary~\ref{corollary_degraded_common}}

\label{proof_of_corollary_degraded_common}

We need to show that (\ref{C_PLNMWC}) reduces to (\ref{C_PDMWC})
for parallel degraded multi-receiver wiretap channels. Consider
the $k$th user. If in the $l$th sub-channel, eavesdropper receives
a degraded version of the $k$th user's observation, i.e., if the
Markov chain in (\ref{degraded_1}) is satisfied, then we have
\begin{align}
\left[I(X_l;Y_{kl})-I(X_l;Z_l)\right]^+=I(X_l;Y_{kl},Z_l)-I(X_l;Z_l)=I(X_l;Y_{kl}|Z_l)
\end{align}
where the first equality is due to the Markov chain in
(\ref{degraded_1}). On the other hand, if in the $l$th
sub-channel, the $k$th user receives a degraded version of the
eavesdropper's observation, i.e., if the Markov chain in
(\ref{degraded_2}) is satisfied, then we have
\begin{align}
\left[I(X_l;Y_{kl})-I(X_l;Z_l)\right]^+=\left[I(X_l;Y_{kl})-I(X_l;Z_l,Y_{kl})\right]^+=0=I(X_l;Y_{kl}|Z_l)
\end{align}
where the first equality is due to the Markov chain in
(\ref{degraded_2}).

\section{Proof of Theorem~\ref{sum_secrecy_capacity}}
\label{proof_of_sum_secrecy_capacity} Achievability of
Theorem~\ref{sum_secrecy_capacity} is a consequence of the
achievability result for wiretap channels in \cite{Korner}. We
provide the converse proof here. We first define the function
$\rho(l)$ which denotes the index of the strongest user in the
$l$th subchannel in the sense that
\begin{align}
I(U;Y_{kl})\leq I(U;Y_{\rho(l)l})
\end{align}
for all  $U\rightarrow X_{l}\rightarrow (Y_{1l},\ldots,Y_{Kl},Z_l)$ and any $k\in\{1,\ldots,K\}$.
Moreover, we define the following shorthand notations
\begin{align}
\tilde{Y}_l^n&=Y_{\rho(l)l}^n,\hspace{3cm} \quad l=1,\ldots,M \\
\tilde{Y}^n&=(\tilde{Y}_1^n,\ldots,\tilde{Y}_M^n)\\
Y_{k}^n&=(Y_{k1}^n,\ldots,Y_{kM}^n),\qquad \qquad k=1,\ldots,K \\
Z^{n}&= (Z_{1}^{n},\ldots,Z_{M}^{n})\\
Y_{k}^{i-1}&=(Y_{k1}^{i-1},\ldots,Y_{kM}^{i-1}),\qquad \quad k=1,\ldots,K \\
Z^{i-1}&= (Z_{1}^{i-1},\ldots,Z_{M}^{i-1})\\
\tilde{Y}_{i+1}^n&=(\tilde{Y}_{1,i+1}^n,\ldots,\tilde{Y}_{M,i+1}^n)\\
Y_k^{l-1}(i)&=(Y_{k1}(i),\ldots,Y_{k,l-1}(i)),\qquad l=1,\ldots,M  \\
Z^{l-1}(i)&=(Z_{1}(i),\ldots,Z_{l-1}(i)),\qquad l=1,\ldots,M  \\
\tilde{Y}_{l+1}^M(i)&=(\tilde{Y}_{l+1}(i),\ldots,\tilde{Y}_M(i)),\qquad l=1,\ldots,M
\end{align}
We first introduce the following lemma.
\begin{Lem}
For the parallel multi-receiver wiretap channel with less
noisiness order, we have
\begin{align}
I(W_k;Y_{k}^n)\leq I(W_k;\tilde{Y}^n),\quad k=1,\ldots,K
\end{align}
\end{Lem}
\begin{proof}
Consecutive uses of Csiszar-Korner identity~\cite{Korner}, as in
Appendix~\ref{proof_of_common_less_noisier_parallel}, yield
\begin{align}
I(W_k;Y_{k}^n)-I(W_k;\tilde{Y}^n)&=\sum_{i=1}^n\sum_{l=1}^M
\left[I(W_k;Y_{kl}(i)|Y_{k}^{i-1},\tilde{Y}_{i+1}^n,Y_k^{l-1}(i),\tilde{Y}_{l+1}^M(i))\right.\nonumber\\
&\quad \left.-I(W_k;\tilde{Y}_{l}(i)|Y_{k}^{i-1},\tilde{Y}_{i+1}^n,Y_k^{l-1}(i),\tilde{Y}_{l+1}^M(i))\right]
\end{align}
where each of the summand is negative, i.e., we have
\begin{align}
I(W_k;Y_{kl}(i)|Y_{k}^{i-1},\tilde{Y}_{i+1}^n,Y_k^{l-1}(i),\tilde{Y}_{l+1}^M(i))-I(W_k;\tilde{Y}_{l}(i)|Y_{k}^{i-1},\tilde{Y}_{i+1}^n,Y_k^{l-1}(i),\tilde{Y}_{l+1}^M(i))
\leq 0
\end{align}
because $\tilde{Y}_{l}(i)$ is the observation of the strongest
user in the $l$th sub-channel, i.e., its channel is less noisy
with respect to all other users in the $l$th sub-channel. This
concludes the proof of the lemma.
\end{proof}
This lemma implies that
\begin{align}
H(W_k|\tilde{Y}^n)\leq H(W_k|Y_k^n)\leq \epsilon_n
\label{dummy_lemma_implies}
\end{align}
where the second inequality is due to Fano's lemma. Using (\ref{dummy_lemma_implies}), we get
\begin{align}
H(W_1,\ldots,W_K|\tilde{Y}^n)&\leq \sum_{k=1}^K H(W_k|\tilde{Y}^n)\leq K\epsilon_n
\label{dummy_lemma_implies_1}
\end{align}
where the first inequality follows from the fact that conditioning cannot increase
entropy.

We now start the converse proof.
\begin{align}
H(W_1,\ldots,W_K|Z^n)&\leq I(W_1,\ldots,W_K;\tilde{Y}^n)- I(W_1,\ldots,W_K;Z^n)+K\epsilon_n
\label{dummy_lemma_implies_2}\\
&=\sum_{i=1}^n\sum_{l=1}^M
\left[
I(W_1,\ldots,W_K;\tilde{Y}_{l}(i)|Z^{i-1},\tilde{Y}_{i+1}^n,Z^{l-1}(i),\tilde{Y}_{l+1}^M(i))
\right.\nonumber\\
&\quad \left.-
I(W_1,\ldots,W_K;Z_{l}(i)|Z^{i-1},\tilde{Y}_{i+1}^n,Z^{l-1}(i),\tilde{Y}_{l+1}^M(i))\right]+K\epsilon_n
\label{consecutive_CK}
\end{align}
where (\ref{dummy_lemma_implies_2}) is a consequence of
(\ref{dummy_lemma_implies_1}) and (\ref{consecutive_CK}) is
obtained via consecutive uses of the Csiszar-Korner
identity~\cite{Korner} as we did in
Appendix~\ref{proof_of_common_less_noisier_parallel}. We define
the set of indices $\mathcal{S}$ such that for all
$l\in\mathcal{S}$, the strongest user in the $l$th sub-channel has
a less noisy channel with respect to the eavesdropper, i.e., we
have
\begin{align}
I(U;\tilde{Y}_{l}(i))\geq I(U;Z_l(i))
\end{align}
for all $U\rightarrow X_{l}(i)\rightarrow (\tilde{Y}_l(i),Z_l(i))$
and any $l\in\mathcal{S}$. Thus, we can further bound
(\ref{consecutive_CK}) as follows,
\begin{align}
H(W_1,\ldots,W_K|Z^n)&\leq \sum_{i=1}^n\sum_{l\in\mathcal{S}}
\left[
I(W_1,\ldots,W_K;\tilde{Y}_{l}(i)|Z^{i-1},\tilde{Y}_{i+1}^n,Z^{l-1}(i),\tilde{Y}_{l+1}^M(i))
\right.\nonumber\\
&\quad \left.-
I(W_1,\ldots,W_K;Z_{l}(i)|Z^{i-1},\tilde{Y}_{i+1}^n,Z^{l-1}(i),\tilde{Y}_{l+1}^M(i))\right]+K\epsilon_n
\label{drop_negatives_2}\\
&\leq \sum_{i=1}^n\sum_{l\in\mathcal{S}} \left[
I(W_1,\ldots,W_K,Z^{i-1},\tilde{Y}_{i+1}^n,Z^{l-1}(i),\tilde{Y}_{l+1}^M(i);\tilde{Y}_{l}(i))
\right.\nonumber\\
&\quad \left.-
I(W_1,\ldots,W_K,Z^{i-1},\tilde{Y}_{i+1}^n,Z^{l-1}(i),\tilde{Y}_{l+1}^M(i);Z_{l}(i))\right]+K\epsilon_n
\label{less_noisy_imply_1}\\
&\leq \sum_{i=1}^n\sum_{l\in\mathcal{S}} \left[
I(X_l(i),W_1,\ldots,W_K,Z^{i-1},\tilde{Y}_{i+1}^n,Z^{l-1}(i),\tilde{Y}_{l+1}^M(i);\tilde{Y}_{l}(i))
\right.\nonumber\\
&\quad \left.-
I(X_l(i),W_1,\ldots,W_K,Z^{i-1},\tilde{Y}_{i+1}^n,Z^{l-1}(i),\tilde{Y}_{l+1}^M(i);Z_{l}(i))\right]+K\epsilon_n
\label{less_noisy_imply_2}\\
&= \sum_{i=1}^n\sum_{l\in\mathcal{S}} \left[
I(X_l(i);\tilde{Y}_{l}(i)) -I(X_l(i);Z_{l}(i))\right]+K\epsilon_n
\label{memoryless_channel_11}
\end{align}
where (\ref{drop_negatives_2}) is obtained by dropping the
negative terms,
(\ref{less_noisy_imply_1})-(\ref{less_noisy_imply_2}) are due to
the following inequalities
\begin{align}
I(Z^{i-1},\tilde{Y}_{i+1}^n,Z^{l-1}(i),\tilde{Y}_{l+1}^M(i);\tilde{Y}_{l}(i))&\geq
I(Z^{i-1},\tilde{Y}_{i+1}^n,Z^{l-1}(i),\tilde{Y}_{l+1}^M(i);Z_{l}(i))
\end{align}
\begin{align}
\lefteqn{I(X_l(i);\tilde{Y}_{l}(i)|W_1,\ldots,W_K,Z^{i-1},\tilde{Y}_{i+1}^n,Z^{l-1}(i),\tilde{Y}_{l+1}^M(i))
\geq}\nonumber\\
&\hspace{4cm}I(X_l(i);Z_{l}(i)|W_1,\ldots,W_K,Z^{i-1},\tilde{Y}_{i+1}^n,Z^{l-1}(i),\tilde{Y}_{l+1}^M(i))
\end{align}
which come from the fact that for any $l\in\mathcal{S}$, the
strongest user in the $l$th sub-channel has a less noisy channel
with respect to the eavesdropper. Finally, we get
(\ref{memoryless_channel_11}) using the following Markov chain
\begin{align}
(W_1,\ldots,W_K,Z^{i-1},\tilde{Y}_{i+1}^n,Z^{l-1}(i),\tilde{Y}_{l+1}^M(i))\rightarrow
X_{l}(i)\rightarrow (\tilde{Y}_{l},Z_l(i))
\end{align}
which is a consequence of the facts that channel is memoryless,
and the sub-channels are independent.

\section{Proofs of Theorems~\ref{Theorem_product_wiretap_channel}~and~\ref{Theorem_D_PWC_M}}

\subsection{Proof of Theorem~\ref{Theorem_product_wiretap_channel}}
\label{proof_of_product_wiretap_channel}

We prove Theorem~\ref{Theorem_product_wiretap_channel} in two
parts, first achievability and then converse. Throughout the
proof, we use the shorthand notations $Y_1^n=(Y_{11}^n,Y_{12}^n)$,
$Y_2^n=(Y_{21}^n,Y_{22}^n)$, $Z_1^n=(Z_{1}^n,Z_{2}^n)$.

\subsubsection{Achievability} To show the achievability of the
region given by
(\ref{product_common_rate_1})-(\ref{product_sum_rate_2}), first we
need to note that the boundary of this region can be decomposed
into three surfaces as follows \cite{Product_Broadcast}.

\begin{itemize}
\item First surface:
\begin{align}
R_0&\leq I(U_2;Y_{12}|Z_2)\label{first_surface_1}\\
R_2&\leq I(X_2;Y_{22}|U_2,Z_2)\label{first_surface_2}\\
R_0+R_1&\leq
I(X_1;Y_{11}|Z_1)+I(U_2;Y_{12}|Z_2)\label{first_surface_3},\quad
U_1=\phi
\end{align}
\item Second surface:
\begin{align}
R_0&\leq I(U_1;Y_{21}|Z_1)\label{second_surface_1}\\
R_1&\leq I(X_1;Y_{11}|U_1,Z_1)\label{second_surface_2}\\
R_0+R_2 &\leq I(X_2;Y_{22}|Z_2)+I(U_1;Y_{21}|Z_1)
\label{second_surface_3},\quad U_2=\phi
\end{align}
\item Third surface:
\begin{align}
R_0&\leq I(U_1;Y_{11}|Z_1)+I(U_2;Y_{12}|Z_2)\label{third_surface_1}\\
R_0&\leq I(U_1;Y_{21}|Z_1)+I(U_2;Y_{22}|Z_2)\label{third_surface_2}\\
R_1&\leq I(X_1;Y_{11}|U_1,Z_1)\label{third_surface_3}\\
R_2&\leq I(X_2;Y_{22}|U_2,Z_2)\label{third_surface_4}
\end{align}
\end{itemize}
We now show the achievability of these regions separately. Start
with the first region.

\begin{Prop}
The region defined by
(\ref{first_surface_1})-(\ref{first_surface_3}) is achievable.
\end{Prop}
\begin{proof}
Fix the probability distribution
\begin{align}
p(x_1)p(u_2)p(x_2|u_2)p(y_1,y_2,z|x)
\end{align}
\noindent\underline{\textbf{Codebook generation:}}
\begin{itemize}
\item Split the private message rate of user 1 as
$R_1=R_{11}+R_{12}$.

\item Generate $2^{n(R_{11}+\tilde{R}_{11})}$ length-$n$ sequences
$\bx_1$ through $p(\bx_1)=\prod_{i=1}^n p(x_{1,i})$ and index them
as $\bx_1(w_{11},\tilde{w}_{11})$ where
$w_{11}\in\left\{1,\ldots,2^{nR_{11}}\right\}$ and
$\tilde{w}_{11}\in\left\{1,\ldots,2^{n\tilde{R}_{11}}\right\}$.

\item Generate $2^{n(R_0+R_{12}+\tilde{R}_{12})}$ length-$n$
sequences $\bu_2$ through $p(\bu_2)=\prod_{i=1}^n p(u_{2,i})$ and
index them as $\bu_2(w_0,w_{12},\tilde{w}_{12})$ where
$w_0\in\left\{1,\ldots,2^{nR_{0}}\right\}$,
$w_{12}\in\left\{1,\ldots,2^{nR_{12}}\right\}$ and
$\tilde{w}_{12}\in\left\{1,\ldots,2^{n\tilde{R}_{12}}\right\}$.

\item For each $\bu_2$, generate $2^{n(R_2+\tilde{R}_2)}$
length-$n$ sequences $\bx_2$ through $p(\bx_2|\bu_2)=\prod_{i=1}^n
p(x_{2,i}|u_{2,i})$ and index them as $\bx_2
(w_2,\tilde{w}_2,w_0,w_{12},\tilde{w}_{12})$ where
$w_2\in\left\{1,\ldots,2^{nR_2}\right\}$,
$\tilde{w}_2\in\left\{1,\ldots,2^{n\tilde{R}_2}\right\}$.

\item Furthermore, set the confusion message rates as follows.
\begin{align}
\tilde{R}_{11} &= I(X_1;Z_1)\label{first_surface_dummy_rate_1}\\
\tilde{R}_{12} & =I(U_2;Z_2) \label{first_surface_dummy_rate_2} \\
\tilde{R}_{2}  &=I(X_2;Z_2|U_2) \label{first_surface_dummy_rate_3}
\end{align}
\end{itemize}

\vspace{0.5cm} \noindent\underline{\textbf{Encoding:}}

\vspace{0.25cm}  If $(w_0,w_{11},w_{12},w_2)$ is the message to be
transmitted, then the receiver randomly picks
$\left(\tilde{w}_{11},\tilde{w}_{12},\tilde{w}_2\right)$ and sends
the corresponding codewords through each channel.

\vspace{0.5cm} \noindent\underline{\textbf{Decoding:}}

\vspace{0.25cm} It is straightforward to see that if the following
conditions are satisfied, then both users can decode the messages
directed to themselves with vanishingly small error probability.
\begin{align}
R_0+\tilde{R}_{12}+R_{12}&\leq I(U_2;Y_{12})\\
R_{11}+\tilde{R}_{11}&\leq I(X_1;Y_{11})\\
R_2+\tilde{R}_2 & \leq I(X_2;Y_{22}|U_2)
\end{align}
After eliminating $R_{11}$ and $R_{12}$ and plugging the values of
$\tilde{R}_{11},\tilde{R}_{12},\tilde{R}_2$, we can reach the
following conditions,
\begin{align}
R_0&\leq I(U_2;Y_{12}|Z_2)\\
R_2&\leq I(X_2;Y_{22}|U_2,Z_2)\\
R_0+R_1&\leq I(X_1;Y_{11}|Z_1)+I(U_2;Y_{12}|Z_2)
\end{align}
where we used the degradedness of the channel. Thus, we only need
to show that this coding scheme satisfies the secrecy constraints.

\vspace{0.5cm} \noindent\underline{\textbf{Equivocation
computation}}:

\vspace{0.1cm} As shown previously in
Lemma~\ref{lemma_sum_rate_sufficient} of
Appendix~\ref{proof_of_less_noisy_multiuser_wiretap}, checking the
sum rate secrecy condition is sufficient.
\begin{align}
\lefteqn{H(W_0,W_1,W_2|Z^n)=H(W_0,W_1,W_2,Z^n)-H(Z^n)}\nonumber \\
&=H(W_0,W_1,W_2,U_2^n,X_2^n,X_1^n,Z^n)-H(U_2^n,X_2^n,X_1^n|W_0,W_1,W_2,Z^n)-H(Z^n)\\
&=H(U_2^n,X_2^n,X_1^n)+H(W_0,W_1,W_2,Z^n|U_2^n,X_2^n,X_1^n)-H(Z^n)\nonumber
\\
&\quad -H(U_2^n,X_2^n,X_1^n|W_0,W_1,W_2,Z^n)\\
&\geq
H(U_2^n,X_2^n,X_1^n)+H(Z^n|U_2^n,X_2^n,X_1^n)-H(Z^n)-H(U_2^n,X_2^n,X_1^n|W_0,W_1,W_2,Z^n)
\label{product_ach_equi_1_step_1}
\end{align}
We treat each term in (\ref{product_ach_equi_1_step_1})
separately. The first term in (\ref{product_ach_equi_1_step_1}) is
\begin{align}
H(U_2^n,X_2^n,X_1^n)&=H(U_2^n,X_2^n)+H(X_1^n)\\
&=
n(R_0+R_{11}+R_2+R_{12}+\tilde{R}_{11}+\tilde{R}_{12}+\tilde{R}_2)
\label{product_ach_equi_1_step_11}
\end{align}
where the first equality is due to the independence of
$(U_2^n,X_2^n)$ and $X_1^n$, and the second equality is due the
fact that both messages and confusion codewords are uniformly
distributed. The second and the third terms in
(\ref{product_ach_equi_1_step_1}) are
\begin{align}
H(Z^n)-H(Z^n|U_2^n,X_2^n,X_1^n)&=
H(Z_1^n,Z_2^n)-H(Z^n|U_2^n,X_2^n,X_1^n)\\
&\leq H(Z_1^n)+H(Z_2^n)-H(Z_1^n,Z_2^n|U_2^n,X_2^n,X_1^n)\\
&=H(Z_1^n)+H(Z_2^n)-H(Z_1^n,Z_2^n|X_2^n,X_1^n)\label{product_ach_markov_1}\\
&= H(Z_1^n)+H(Z_2^n)-H(Z_1^n|X_1^n)-H(Z_2^n|X_2^n)\label{product_ach_markov_2}\\
&=I(X_1^n;Z_1^n)+I(X_2^n;Z_2^n)\\
&\leq nI(X_1;Z_1)+nI(X_2;Z_2)+\gamma_{1,n}+\gamma_{2,n}
\label{memoryless_channel_8}
\end{align}
where the equalities in (\ref{product_ach_markov_1}) and
(\ref{product_ach_markov_2}) are due to the following Markov
chains
\begin{align}
U_2^n & \rightarrow X_2^n \rightarrow
\left(X_1^n,Z_1^n,Z_2^n\right)\\
Z_2^n&\rightarrow X_2^n\rightarrow X_1^n \rightarrow Z_1^n
\end{align}
respectively, and the last inequality in
(\ref{memoryless_channel_8}) can be shown using the technique
devised in~\cite{Wyner}. To bound the last term in
(\ref{product_ach_equi_1_step_1}), assume that the eavesdropper
tries to decode $\left(U_2^n,X_2^n,X_1^n\right)$ using the side
information $W_0,W_1,W_2$ and its observation. Since the rates of
the confusion codewords are selected such that the eavesdropper
can decode them given $W_0=w_0,W_1=w_1,W_2=w_2$ (see
(\ref{first_surface_dummy_rate_1})-(\ref{first_surface_dummy_rate_3})),
using Fano's lemma, we get
\begin{align}
H(U_2^n,X_2^n,X_1^n|W_0,W_1,W_2,Z^n)\leq \epsilon_n
\label{product_ach_equi_1_step_12}
\end{align}
for the third term in (\ref{product_ach_equi_1_step_1}). Plugging
(\ref{product_ach_equi_1_step_11}), (\ref{memoryless_channel_8})
and (\ref{product_ach_equi_1_step_12}) into
(\ref{product_ach_equi_1_step_1}), we get
\begin{align}
H(W_0,W_1,W_2|Z^n) \geq
n(R_0+R_1+R_2)-\epsilon_n-\gamma_{1,n}-\gamma_{2,n}
\end{align}
which completes the proof.
\end{proof}

Achievability of the region defined by
(\ref{second_surface_1})-(\ref{second_surface_3}) follows due to
symmetry. We now show the achievability of the region defined by
(\ref{third_surface_1})-(\ref{third_surface_4}).

\begin{Prop} The region described by
(\ref{third_surface_1})-(\ref{third_surface_4}) is achievable.
\end{Prop}
\begin{proof}
Fix the probability distribution as follows,
\begin{align}
p(u_1)p(x_1|u_1)p(u_2)p(x_2|u_2)p(y_1,y_2,z|x)
\end{align}

\noindent \underline{\textbf{Codebook generation:}}

\begin{itemize}

\item Generate $2^{n(R_0+\tilde{R}_{01})}$ length-$n$ sequences
$\bu_1$ through $p(\bu_1)=\prod_{i=1}^{n}p(u_{1,i})$ and index
them as $\bu_1 (w_0,\tilde{w}_{01})$ where
$w_0\in\left\{1,\ldots,2^{nR_0}\right\}$,
$\tilde{w}_{01}\in\left\{1,\ldots,2^{n\tilde{R}_{01}}\right\}$.

\item For each $\bu_1$, generate $2^{n(R_1+\tilde{R}_1)}$ $\bx_1
(w_0,\tilde{w}_{01},w_1,\tilde{w}_1)$ length-$n$ sequences $\bx_1$
through $p(\bx_1)=\prod_{i=1}^n p(x_{1,i}|u_{1,i})$ where
$w_1\in\left\{1,\ldots,2^{nR_1}\right\}$,
$\tilde{w}_1\in\left\{1,\ldots,2^{n\tilde{R}_1}\right\}$.

\item Generate $2^{n(R_0+\tilde{R}_{02})}$ length-$n$ sequences
$\bu_2$ through $p(\bu_2)=\prod_{i=1}^{n}p(u_{2,i})$ and index
them as $\bu_2 (w_0,\tilde{w}_{02})$ where
$w_0\in\left\{1,\ldots,2^{nR_0}\right\}$,
$\tilde{w}_{02}\in\left\{1,\ldots,2^{n\tilde{R}_{02}}\right\}$.

\item For each $\bu_2$, generate $2^{n(R_2+\tilde{R}_2)}$ $\bx_2
(w_0,\tilde{w}_{02},w_2,\tilde{w}_2)$ length-$n$ sequences $\bx_2$
through $p(\bx_2)=\prod_{i=1}^n p(x_{2,i}|u_{2,i})$ where
$w_2\in\left\{1,\ldots,2^{nR_2}\right\}$,
$\tilde{w}_2\in\left\{1,\ldots,2^{n\tilde{R}_2}\right\}$.

\item Moreover, set the rates of confusion messages as follows,
\begin{align}
\tilde{R}_{01}&=I(U_1;Z_1)\label{third_surface_dummy_rate_1}\\
\tilde{R}_{02}&=I(U_2;Z_2)\label{third_surface_dummy_rate_2}\\
\tilde{R}_{1}&=I(X_1;Z_1|U_1)\label{third_surface_dummy_rate_3}\\
\tilde{R}_{2}&=I(X_2;Z_2|U_2)\label{third_surface_dummy_rate_4}
\end{align}

\end{itemize}

\vspace{0.5cm} \noindent \underline{\textbf{Encoding:}}

\vspace{0.25cm} Assume that the messages to be transmitted are
$\left(w_0,w_1,w_2\right)$. Then, after randomly picking the tuple
$(\tilde{w}_{01},\tilde{w}_{02},\tilde{w}_1,\tilde{w}_2)$,
corresponding codewords are sent.

\vspace{0.5cm} \noindent \underline{\textbf{Decoding:}}

\vspace{0.25cm} Users decode $w_0$ using their both observations.
If $w_0$ is the only message that satisfies
\begin{align}
E_{i1}^{w_0}=\left\{ \exists
\tilde{w}_{01}:\left(\bu_{1}(w_0,\tilde{w}_{01}),\by_{i1}\right)\in
A_{\epsilon}^n\right\}\\
E_{i2}^{w_0}=\left\{\exists
\tilde{w}_{02}:\left(\bu_{2}(w_0,\tilde{w}_{02}),\by_{i2}\right)\in
A_{\epsilon}^n\right\}
\end{align}
simultaneously for user $i$, $w_0$ is declared to be transmitted.
Assume $w_0=1$ is transmitted. The error probability for user $i$
can be bounded as
\begin{align}
\label{sum_unmatch_third_surface_start} \Pr\left(E_i\right) \leq
\Pr \left(\left(E_{i1}^1,E_{i2}^1\right)
^c\right)+\sum_{j=2}^{2^{nR_0}} \Pr \left(E_{i1}^j,E_{i2}^j\right)
\end{align}
using the union bound. Let us consider the following
\begin{align}
\Pr \left(E_{i1}^j\right)&=\Pr\left(\exists
\tilde{w}_{01}:\left(\bu_{1}(j,\tilde{w}_{01}),\by_{i1}\right)\in
A_{\epsilon}^n  \right)\\
&\leq \sum_{\forall \tilde{w}_{01}}
\Pr\left(\left(\bu_{1}(j,\tilde{w}_{01}),\by_{i1}\right)\in
A_{\epsilon}^n  \right)\\
&\leq 2^{n\tilde{R}_{01}}2^{-n(I(U_1;Y_{i1})-\epsilon_n)} \\
&=2^{n(\tilde{R}_{01}-I(U_1;Y_{i1})+\epsilon_n)}
\end{align}
Similarly, we have
\begin{align}
\Pr \left(E_{i2}^j\right) &\leq
2^{n(\tilde{R}_{02}-I(U_2;Y_{i2})+\epsilon_n)}
\end{align}
Thus, the probability of declaring that the $j$th message was
transmitted can be bounded as
\begin{align}
 \Pr \left(E_{i1}^j,E_{i2}^j\right)&=\Pr \left(E_{i1}^j\right)\times
 \Pr\left(E_{i2}^j\right)\\
 &\leq 2^{n(\tilde{R}_{01}-I(U_1;Y_{i1})+\epsilon_n)}\times
2^{n(\tilde{R}_{02}-I(U_2;Y_{i2})+\epsilon_n)}\\
&= 2^{n(\tilde{R}_{01}-I(U_1;Y_{i1})+
\tilde{R}_{02}-I(U_2;Y_{i2})+2\epsilon_n)}
\end{align}
where the first equality is due to the independence of
sub-channels and codebooks used for each channel. Therefore, error
probability can be bounded as
\begin{align}
\Pr\left(E_i\right)& \leq \epsilon_n
+\sum_{j=2}^{2^{nR_0}}2^{n(\tilde{R}_{01}-I(U_1;Y_{i1})+
\tilde{R}_{02}-I(U_2;Y_{i2})+2\epsilon_n)} \\
&=\epsilon_n + 2^{n(R_0+\tilde{R}_{01}-I(U_1;Y_{i1})+
\tilde{R}_{02}-I(U_2;Y_{i2})+2\epsilon_n)}
\end{align}
which vanishes if the following are satisfied,
\begin{align}
R_0+\tilde{R}_{01}+\tilde{R}_{02}\leq
I(U_1;Y_{i1})+I(U_2;Y_{i2}),\quad i=1,2
\label{third_surface_12_derive}
\end{align}
After decoding the common message, both users decode their private
messages if the rates satisfy
\begin{align}
R_1+\tilde{R}_1&\leq I(X_1;Y_{11}|U_1)\label{third_surface_3_derive}\\
R_2+\tilde{R}_2&\leq I(X_2;Y_{22}|U_2)
\label{third_surface_4_derive}
\end{align}
After plugging the values of
$\tilde{R}_{01},\tilde{R}_{02},\tilde{R}_1,\tilde{R}_2$ given by
(\ref{third_surface_dummy_rate_1})-(\ref{third_surface_dummy_rate_4})
into
(\ref{third_surface_12_derive})-(\ref{third_surface_4_derive}),
one can recover the region described by
(\ref{third_surface_1})-(\ref{third_surface_4}) using the
degradedness of the channel.

\vspace{0.5cm} \noindent \underline{\textbf{Equivocation
calculation:}}

\vspace{0.25cm} It is sufficient to check the sum rate constraint,
\begin{align}
\lefteqn{H(W_0,W_1,W_2|Z^n)=H(W_0,W_1,W_2,Z^n)-H(Z^n)}\\
&=H(U_1^n,U_2^n,X_1^n,X_2^n,W_0,W_1,W_2,Z^n)-H(U_1^n,U_2^n,X_1^n,X_2^n|W_0,W_1,W_2,Z^n)\nonumber
\\
&\quad -H(Z^n)\\
&=H(U_1^n,U_2^n,X_1^n,X_2^n)+H(W_0,W_1,W_2,Z^n|U_1^n,U_2^n,X_1^n,X_2^n)-H(Z^n)\nonumber
\\
&\quad -H(U_1^n,U_2^n,X_1^n,X_2^n|W_0,W_1,W_2,Z^n)\\
&\geq
H(U_1^n,U_2^n,X_1^n,X_2^n)+H(Z^n|U_1^n,U_2^n,X_1^n,X_2^n)-H(Z^n)\nonumber
\\
&\quad
-H(U_1^n,U_2^n,X_1^n,X_2^n|W_0,W_1,W_2,Z^n)\label{product_ach_equi_2_step_1}
\end{align}
where each term will be treated separately. The first term is
\begin{align}
H(U_1^n,U_2^n,X_1^n,X_2^n)&=H(U_1^n,U_2^n)+H(X_1^n|U_1^n,U_2^n)+H(X_1^n|U_1^n,U_2^n)\\
&=n(R_0+R_1+R_2+\tilde{R}_{01}+\tilde{R}_{02}+\tilde{R}_1+\tilde{R}_2)
\label{product_ach_equi_2_step_11}
\end{align}
where we first use the fact that $X_1^n$ and $X_2^n$ are
independent given $(U_1^n,U_2^n)$ and secondly, we use the fact
that messages are uniformly distributed. The second and third term
of (\ref{product_ach_equi_2_step_1}) are
\begin{align}
H(Z^n)-H(Z^n|U_1^n,U_2^n,X_1^n,X_2^n)&=
H(Z_1^n,Z_2^n)-H(Z_1^n|X_1^n)-H(Z_1^n|X_2^n)\\
&\leq H(Z_1^n)+H(Z_2^n)-H(Z_1^n|X_1^n)-H(Z_1^n|X_2^n)\\
&=I(X_1^n;Z_1^n)+I(X_2^n;Z_2^n)\\
&\leq nI(X_1;Z_1)+nI(X_2;Z_2)+\gamma_{1,n}+\gamma_{2,n}
\label{product_ach_equi_2_step_12}
\end{align}
where the first equality is due to the independence of the
sub-channels. We now consider the last term of
(\ref{product_ach_equi_2_step_1}) for which assume that
eavesdropper tries to decode
$\left(U_1^n,U_2^n,X_1^n,X_2^n\right)$ using the side information
$(W_0,W_1,W_2)$ and its observation. Since the rates of the
confusion messages are selected to ensure that the eavesdropper
can decode $\left(U_1^n,U_2^n,X_1^n,X_2^n\right)$ given
$(W_0=w_0,W_1=w_1,W_2=w_2)$ (see
(\ref{third_surface_dummy_rate_1})-(\ref{third_surface_dummy_rate_4})),
using Fano's lemma we have
\begin{align}
H(U_1^n,U_2^n,X_1^n,X_2^n|W_0,W_1,W_2,Z^n)\leq \epsilon_n
\label{product_ach_equi_2_step_13}
\end{align}
Plugging (\ref{product_ach_equi_2_step_11}),
(\ref{product_ach_equi_2_step_12}) and
(\ref{product_ach_equi_2_step_13}) into
(\ref{product_ach_equi_2_step_1}), we have
\begin{align}
H(W_0,W_1,W_2|Z^n)\geq
n(R_0+R_1+R_2)-\epsilon_n-\gamma_{1,n}-\gamma_{2,n}
\end{align}
which concludes the proof.
\end{proof}

\subsubsection{Converse}

First let us define the following auxiliary random variables,
\begin{align}
U_{1,i}& = W_0W_2Y_{12}^n Y_{11}^{i-1} Z_{1,i+1}^n \\
U_{2,i}&= W_0 W_1 Y_{21}^n Y_{22}^{i-1} Z_{2,i+1}^n
\end{align}
which satisfy the following Markov chains
\begin{align}
U_{1,i} \rightarrow X_{1,i} \rightarrow
\left(Y_{11,i},Y_{21,i},Z_{1,i}\right)\\
U_{2,i} \rightarrow X_{2,i} \rightarrow
\left(Y_{12,i},Y_{22,i},Z_{2,i}\right)
\end{align}
We remark that although $U_{1,i}$ and $U_{2,i}$ are correlated, at
the end of the proof, it will turn out that selection of them as
independent will yield the same region. We start with the common
message rate,
\begin{align}
H(W_0|Z^n)&= H(W_0)-I(W_0;Z^n) \\
&\leq I(W_0;Y_1^n)-I(W_0;Z^n) +\epsilon_n \label{Fano_4} \\
& = I(W_0;Y_1^n|Z^n) +\epsilon_n \label{product_degraded_imply_1}
\\
& = I(W_0;Y_{12}^n|Z^n)+I(W_0;Y_{11}^n|Y_{12}^n,Z^n) +\epsilon_n \\
& \leq I(W_0,W_1;Y_{12}^n|Z^n)
+I(W_0,W_2;Y_{11}^n|Y_{12}^n,Z^n)+\epsilon_n
\label{converse_product_common_step_1}
\end{align}
where (\ref{Fano_4}) is due to Fano's lemma, equality in
(\ref{product_degraded_imply_1}) is due to the fact that the
eavesdropper's channel is degraded with respect to the first
user's channel. We bound each term in
(\ref{converse_product_common_step_1}) separately. First term is
\begin{align}
I(W_0,W_1;Y_{12}^n|Z^n)&= \sum_{i=1}^n
I(W_0,W_1;Y_{12,i}|Y_{12}^{i-1},Z_1^n,Z_2^n)\label{converse_common_use_1_start} \\
&=\sum_{i=1}^n H(Y_{12,i}|Y_{12}^{i-1},Z_1^n,Z_2^n)
-H(Y_{12,i}|Y_{12}^{i-1},Z_1^n,Z_2^n,W_0,W_1)\\
&\leq \sum_{i=1}^n H(Y_{12,i}|Z_{2,i})
-H(Y_{12,i}|Y_{12}^{i-1},Z_1^n,Z_2^n,W_0,W_1,Y_{21}^n,Y_{22}^{i-1})
\label{conditioning_cannot_increase_1}\\
& = \sum_{i=1}^n H(Y_{12,i}|Z_{2,i})
-H(Y_{12,i}|W_0,W_1,Y_{21}^n,Y_{22}^{i-1},Z_{2,i+1}^n,Z_{2,i})
\label{conv_prod_degraded_imply}\\
&=\sum_{i=1}^n I(U_{2,i};Y_{12,i}|Z_{2,i})
\label{converse_product_common_step_11}
\end{align}
where (\ref{conditioning_cannot_increase_1}) follows from the fact
that conditioning cannot increase entropy and the equality in
(\ref{conv_prod_degraded_imply}) is due to the following Markov
chains
\begin{align}
Z_1^n&\rightarrow Y_{21}^n \rightarrow
\left(W_0,W_1,Y_{22}^n,Z_2^n,Y_{12}^n\right)\\
Y_{12}^{i-1}Z_2^{i-1} & \rightarrow Y_{22}^{i-1}\rightarrow
\left(W_0,W_1,Y_{21}^n,Y_{12,i},Z_{2,i}^n,Z_1^n\right)
\end{align}
both of which are due to the fact that sub-channels are
independent, memoryless and degraded. We now consider the second
term in (\ref{converse_product_common_step_1}),
\begin{align}
I(W_0,W_2;Y_{11}^n|Y_{12}^n,Z^n)&=\sum_{i=1}^n
I(W_0,W_2;Y_{11,i}|Y_{12}^n,Z_1^n,Z_2^n,Y_{11}^{i-1}) \label{converse_common_use_2_start} \\
&=\sum_{i=1}^n
I(W_0,W_2;Y_{11,i}|Y_{12}^n,Y_{11}^{i-1},Z_{1,i+1}^n,Z_{1,i}) \label{conv_prod_degraded_imply_1}\\
&\leq \sum_{i=1}^n
I(W_0,W_2,Y_{12}^n,Y_{11}^{i-1},Z_{1,i+1}^n;Y_{11,i}|Z_{1,i})\\
&= \sum_{i=1}^n I(U_{1,i};Y_{11,i}|Z_{1,i})
\label{converse_product_common_step_12}
\end{align}
where (\ref{conv_prod_degraded_imply_1}) follows from the
following Markov chains
\begin{align}
Z_2^n &\rightarrow Y_{12}^n\rightarrow
\left(W_0,W_2,Y_{11}^{i-1},Z_1^n,Y_{11,i}\right) \\
Z_1^{i-1}&\rightarrow Y_{11}^{i-1}\rightarrow
(W_0,W_2,Y_{12}^n,Z_{1,i+1}^n,Z_{1,i},Y_{11,i})
\end{align}
both of which are due to the fact that sub-channels are
independent, memoryless and degraded. Plugging
(\ref{converse_product_common_step_11}) and
(\ref{converse_product_common_step_12}) into
(\ref{converse_product_common_step_1}), we get the following outer
bound on the common rate.
\begin{align}
H(W_0|Z^n)&\leq \sum_{i=1}^n I(U_{2,i};Y_{12,i}|Z_{2,i}) +
\sum_{i=1}^n I(U_{1,i};Y_{11,i}|Z_{1,i}) +\epsilon_n
\label{product_outer_1}
\end{align}
Using the same analysis on the second user, we can obtain the
following outer bound on the common rate as well.
\begin{align}
H(W_0|Z^n)&\leq \sum_{i=1}^n I(U_{2,i};Y_{22,i}|Z_{2,i}) +
\sum_{i=1}^n I(U_{1,i};Y_{21,i}|Z_{1,i}) +\epsilon_n
\label{product_outer_2}
\end{align}
We now bound the sum of independent and common message rates for
each user,
\begin{align}
H(W_0,W_1|Z^n)& \leq I(W_0,W_1;Y_1^n)-I(W_0,W_1;Z^n)+\epsilon_n \label{Fano_5}\\
& = I(W_0,W_1;Y_1^n|Z^n)+\epsilon_n \label{product_degraded_imply_2}\\
&= I(W_0,W_1;Y_{11}^n,Y_{12}^n|Z^n) +\epsilon_n \\
& =
I(W_0,W_1;Y_{12}^n|Z^n)+I(W_0,W_1;Y_{11}^n|Y_{12}^n,Z^n)+\epsilon_n
\label{converse_product_common_private_step_1}
\end{align}
where (\ref{Fano_5}) is due to Fano's lemma,
(\ref{product_degraded_imply_2}) is due to the fact that the
eavesdropper's channel is degraded with respect to the first
user's channel. Using (\ref{converse_product_common_step_11}), the
first term in (\ref{converse_product_common_private_step_1}) can
be bounded as
\begin{align}
I(W_0,W_1;Y_{12}^n|Z^n)\leq \sum_{i=1}^n
I(U_{2,i};Y_{12,i}|Z_{2,i})
\label{converse_product_common_private_step_11}
\end{align}
Thus, we only need to bound the second term of
(\ref{converse_product_common_private_step_1}),
\begin{align}
I(W_0,W_1;Y_{11}^n|Y_{12}^n,Z^n)&=H(Y_{11}^n|Y_{12}^n,Z_1^n,Z_2^n)-H(Y_{11}^n|Y_{12}^n,Z_1^n,Z_2^n,W_0,W_1)
\label{converse_common_use_3_start}\\
&\leq
H(Y_{11}^n|Z_1^n)-H(Y_{11}^n|Y_{12}^n,Z_1^n,Z_2^n,W_0,W_1,X_1^n)
\label{conditioning_cannot_increase_2} \\
&= H(Y_{11}^n|Z_1^n)-H(Y_{11}^n|Z_1^n,X_1^n)
\label{memoryless_channel_3} \\
& = I(X_1^n;Y_{11}^n|Z_1^n) \\
&\leq \sum_{i=1}^n H(Y_{11,i}|Z_{1,i})-H(Y_{11,i}|Z_1^n,X_1^n,Y_{11}^{i-1})\label{conditioning_cannot_increase_6}\\
&= \sum_{i=1}^n H(Y_{11,i}|Z_{1,i})-H(Y_{11,i}|Z_{1,i},X_{1,i})\label{memoryless_channel_9x}\\
&= \sum_{i=1}^{n}I(X_{1,i};Y_{11,i}|Z_{1,i})
\label{converse_product_common_private_step_12}
\end{align}
where (\ref{conditioning_cannot_increase_2}) is due to the fact
that conditioning cannot increase entropy,
(\ref{memoryless_channel_3}) is due to the following Markov chain
\begin{align}
\left(Y_{11}^n,Z_1^n\right) \rightarrow X_1^n \rightarrow
\left(Y_{12}^n,Z_2^n,W_0,W_1\right)
\end{align}
and (\ref{conditioning_cannot_increase_6}) follows from the fact
that conditioning cannot increase entropy. Finally,
(\ref{memoryless_channel_9x}) is due to the fact that each
sub-channel is memoryless. Hence, plugging
(\ref{converse_product_common_private_step_11}) and
(\ref{converse_product_common_private_step_12}) into
(\ref{converse_product_common_private_step_1}), we get the
following outer bound.
\begin{align}
H(W_0,W_1|Z^n)& \leq
\sum_{i=1}^{n}I(X_{1,i};Y_{11,i}|Z_{1,i})+\sum_{i=1}^n
I(U_{2,i};Y_{12,i}|Z_{2,i})+\epsilon_n \label{product_outer_3}
\end{align}
Similarly, for the second user, we can get the following outer
bound,
\begin{align}
H(W_0,W_2|Z^n)& \leq
\sum_{i=1}^{n}I(X_{2,i};Y_{22,i}|Z_{2,i})+\sum_{i=1}^n
I(U_{1,i};Y_{21,i}|Z_{1,i})+\epsilon_n \label{product_outer_4}
\end{align}
We now bound the sum rates to conclude the converse,
\begin{align}
\lefteqn{H(W_0,W_1,W_2|Z^n)= H(W_0,W_1,W_2)-I(W_0,W_1,W_2;Z^n)}&\\
&\leq
I(W_0,W_1;Y_1^n)+I(W_2;Y_2^n|W_0,W_1)-I(W_0,W_1,W_2;Z^n)+\epsilon_n
\label{Fano_6}
\\
& = I(W_0,W_1;Y_1^n|Z^n)+I(W_2;Y_2^n|W_0,W_1,Z^n)+\epsilon_n
\label{product_degraded_imply_3}\\
&= I(W_0,W_1;Y_{12}^n|Z^n)+I(W_0,W_1;Y_{11}^n|Z^n,Y_{12}^n)
+I(W_2;Y_{21}^n|W_0,W_1,Z^n)\nonumber\\
&\quad +I(W_2;Y_{22}^n|W_0,W_1,Z^n,Y_{21}^n)+\epsilon_n \label{converse_product_add_subtract_from}\\
&=I(W_0,W_1,Y_{21}^n;Y_{12}^n|Z^n)-I(Y_{21}^n;Y_{12}^n|W_0,W_1,Z^n)+I(W_0,W_1;Y_{11}^n|Z^n,Y_{12}^n)\nonumber\\
&\quad
+I(W_2;Y_{21}^n|W_0,W_1,Z^n)+I(W_2;Y_{22}^n|W_0,W_1,Z^n,Y_{21}^n)+\epsilon_n
\label{converse_product_add_subtract_results}
\\
&= \mbox{S}_1-\mbox{S}_2+\mbox{S}_3+\mbox{S}_4+\mbox{S}_5
\label{converse_product_sum_step_1}
\end{align}
where (\ref{Fano_6}) follows from Fano's lemma,
(\ref{product_degraded_imply_3}) is due to the fact that the
eavesdropper's channel is degraded with respect to both users'
channels, (\ref{converse_product_add_subtract_results}) is
obtained by adding and subtracting $\mbox{S}_2$ from the first
term of (\ref{converse_product_add_subtract_from}). Now, we
proceed as follows.
\begin{align}
\mbox{S}_4-\mbox{S}_2&=
I(W_2;Y_{21}^n|W_0,W_1,Z^n)-I(Y_{21}^n;Y_{12}^n|W_0,W_1,Z^n)\label{converse_common_use_4_start}\\
&\leq
I(W_2,Y_{12}^n;Y_{21}^n|W_0,W_1,Z^n)-I(Y_{21}^n;Y_{12}^n|W_0,W_1,Z^n)\\
&= I(W_2;Y_{21}^n|W_0,W_1,Z^n,Y_{12}^n)
\label{converse_product_sum_step_12}
\end{align}
Adding $\mbox{S}_3$ to (\ref{converse_product_sum_step_12}), we
get
\begin{align}
\mbox{S}_3+\mbox{S}_4-\mbox{S}_2&\leq
I(W_0,W_1;Y_{11}^n|Z^n,Y_{12}^n)+
I(W_2;Y_{21}^n|W_0,W_1,Z^n,Y_{12}^n)\\
&\leq I(W_0,W_1;Y_{11}^n|Z^n,Y_{12}^n)+
I(W_2;Y_{11}^n,Y_{21}^n|W_0,W_1,Z^n,Y_{12}^n)\\
&= I(W_0,W_1;Y_{11}^n|Z^n,Y_{12}^n)+
I(W_2;Y_{11}^n|W_0,W_1,Z^n,Y_{12}^n)\nonumber\\
&\quad +I(W_2;Y_{21}^n|W_0,W_1,Z^n,Y_{12}^n,Y_{11}^n)\\
&=
I(W_0,W_1,W_2;Y_{11}^n|Z^n,Y_{12}^n)+I(W_2;Y_{21}^n|W_0,W_1,Z^n,Y_{12}^n,Y_{11}^n)
\label{converse_product_sum_step_13}
\end{align}
where the second term is zero as we show next,
\begin{align}
\lefteqn{I(W_2;Y_{21}^n|W_0,W_1,Z^n,Y_{12}^n,Y_{11}^n)}\nonumber\\
&=H(W_2|W_0,W_1,Z_1^n,Z_2^n,Y_{12}^n,Y_{11}^n) -H(W_2|W_0,W_1,Z_1^n,Z_2^n,Y_{12}^n,Y_{11}^n,Y_{21}^n)\\
&=H(W_2|W_0,W_1,Y_{12}^n,Y_{11}^n)-H(W_2|W_0,W_1,Y_{12}^n,Y_{11}^n)=0
\end{align}
where we used the following Markov chain
\begin{align}
\left(W_0,W_1,W_2\right)\rightarrow \left(Y_{11}^n,Y_{12}^n\right)
\rightarrow \left(Y_{21}^n,Z_1^n,Z_2^n\right)
\end{align}
which is a consequence of the degradation orders that sub-channels
exhibit. Thus, (\ref{converse_product_sum_step_13}) can be
expressed as
\begin{align}
\mbox{S}_3+\mbox{S}_4-\mbox{S}_2&\leq
I(W_0,W_1,W_2;Y_{11}^n|Z^n,Y_{12}^n) \\
&=I(W_0,W_1,W_2;Y_{11}^n|Z_1^n,Y_{12}^n)
\label{product_degraded_imply_4}\\
&\leq I(X_1^n,W_0,W_1,W_2;Y_{11}^n|Z_1^n,Y_{12}^n)\\
&=
I(X_1^n;Y_{11}^n|Z_1^n,Y_{12}^n)+I(W_0,W_1,W_2;Y_{11}^n|Z_1^n,Y_{12}^n,X_1^n)
\label{converse_product_sum_step_14}
\end{align}
where (\ref{product_degraded_imply_4}) follows from the following
Markov chain
\begin{align}
Z_2^n \rightarrow  Y_{12}^n \rightarrow
\left(W_0,W_1,W_2,Y_{11}^n,Z_1^n\right)
\end{align}
which is due to the degradedness of the channel. Moreover, the
second term in (\ref{converse_product_sum_step_14}) is zero as we
show next,
\begin{align}
\lefteqn{I(W_0,W_1,W_2;Y_{11}^n|Z_1^n,Y_{12}^n,X_1^n)}&
\nonumber\\
&=H(W_0,W_1,W_2|Z_1^n,Y_{12}^n,X_1^n)-H(W_0,W_1,W_2|Z_1^n,Y_{12}^n,X_1^n,Y_{11}^n)\\
&=H(W_0,W_1,W_2|Y_{12}^n,X_1^n)-H(W_0,W_1,W_2|Y_{12}^n,X_1^n)=0\label{memoryless_channel_5}
\end{align}
where (\ref{memoryless_channel_5}) follows from the following
Markov chain
\begin{align}
\left(Y_{11}^n,Z_1^n\right)\rightarrow X_1^n \rightarrow
\left(W_0,W_1,W_2,Y_{12}^n\right)
\end{align}
Thus, (\ref{converse_product_sum_step_14}) turns out to be
\begin{align}
\mbox{S}_3+\mbox{S}_4-\mbox{S}_2&\leq
I(X_1^n;Y_{11}^n|Z_1^n,Y_{12}^n)
\end{align}
which can be further bounded as follows,
\begin{align}
\mbox{S}_3+\mbox{S}_4-\mbox{S}_2&\leq
H(Y_{11}^n|Z_1^n,Y_{12}^n)-H(Y_{11}^n|Z_1^n,Y_{12}^n,X_1^n)\\
&\leq H(Y_{11}^n|Z_1^n)-H(Y_{11}^n|Z_1^n,Y_{12}^n,X_1^n)
\label{conditioning_cannot_increase_3} \\
&= H(Y_{11}^n|Z_1^n)-H(Y_{11}^n|Z_1^n,X_1^n)
\label{memoryless_channel_6} \\
&\leq \sum_{i=1}^n I(X_{1,i};Y_{11,i}|Z_{1,i})
\label{converse_product_sum_step_15}
\end{align}
where (\ref{conditioning_cannot_increase_3}) is due to the fact
that conditioning cannot increase entropy,
(\ref{memoryless_channel_6}) is due to the following Markov chain
\begin{align}
\left(Y_{11}^n,Z_1^n\right) \rightarrow X_1^n \rightarrow Y_{12}^n
\end{align}
Finally, (\ref{converse_product_sum_step_15}) is due to our
previous result in
(\ref{converse_product_common_private_step_12}). We keep bounding
terms in (\ref{converse_product_sum_step_1}),
\begin{align}
\mbox{S}_5&= I(W_2;Y_{22}^n|W_0,W_1,Y_{21}^n,Z_1^n,Z_2^n)\label{converse_common_use_5_start}\\
&=
I(W_2;Y_{22}^n|W_0,W_1,Y_{21}^n,Z_2^n)\label{product_degraded_imply_5}\\
&=\sum_{i=1}^n
I(W_2;Y_{22,i}|W_0,W_1,Y_{21}^n,Z_2^n,Y_{22}^{i-1})\\
&=\sum_{i=1}^n
I(W_2;Y_{22,i}|W_0,W_1,Y_{21}^n,Z_{2,i+1}^n,Y_{22}^{i-1},Z_{2,i})
\label{product_degraded_imply_6}\\
&=\sum_{i=1}^n I(W_2;Y_{22,i}|U_{2,i},Z_{2,i})\\
&\leq \sum_{i=1}^n
H(Y_{22,i}|U_{2,i},Z_{2,i})-H(Y_{22,i}|U_{2,i},Z_{2,i},W_2,X_{2,i})
\label{conditioning_cannot_increase_4}\\
 &\leq \sum_{i=1}^n
I(X_{2,i};Y_{22,i}|U_{2,i},Z_{2,i}) \label{memoryless_channel_7}
\end{align}
where (\ref{product_degraded_imply_5}) and
(\ref{product_degraded_imply_6}) are due to the following Markov
chains
\begin{align}
Z_1^n&\rightarrow Y_{21}^n \rightarrow
\left(W_0,W_1,W_2,Y_{22}^n,Z_2^n\right)\\
Z_2^{i-1} & \rightarrow Y_{22}^{i-1} \rightarrow
\left(W_0,W_1,W_2,Y_{21}^n,Z_{2,i}^n,Y_{22,i}\right)
\end{align}
respectively, (\ref{conditioning_cannot_increase_4}) follows from
that conditioning cannot increase entropy and
(\ref{memoryless_channel_7}) is due to the following Markov chain
\begin{align}
\left(Y_{22,i},Z_{2,i}\right) \rightarrow X_{2,i} \rightarrow
\left(W_2,U_{2,i}\right)
\end{align}
which is a consequence of the fact that each sub-channel is
memoryless. Thus, we only need to bound $\mbox{S}_1$ in
(\ref{converse_product_sum_step_1}) to reach the outer bound for
the sum secrecy rate,
\begin{align}
\mbox{S}_1&= I(W_0,W_1,Y_{21}^n;Y_{12}^n|Z^n) \label{converse_common_use_6_start} \\
&=\sum_{i=1}^n I(W_0,W_1,Y_{21}^n;Y_{12,i}|Z_1^n,Z_2^n,Y_{12}^{i-1})\\
&\leq \sum_{i=1}^n
H(Y_{12,i}|Z_{2,i})-H(Y_{12,i}|Z_1^n,Z_2^n,Y_{12}^{i-1},W_0,W_1,Y_{21}^n,Y_{22}^{i-1})
\label{conditioning_cannot_increase_5}\\
&= \sum_{i=1}^n
H(Y_{12,i}|Z_{2,i})-H(Y_{12,i}|Z_2^n,Y_{12}^{i-1},W_0,W_1,Y_{21}^n,Y_{22}^{i-1})
\label{product_degraded_imply_7} \\
&= \sum_{i=1}^n
H(Y_{12,i}|Z_{2,i})-H(Y_{12,i}|W_0,W_1,Y_{21}^n,Y_{22}^{i-1},Z_{2,i+1}^n,Z_{2,i})
\label{product_degraded_imply_8} \\
&=\sum_{i=1}^n I(U_{2,i};Y_{12,i}|Z_{2,i})
\label{converse_product_sum_step_16}
\end{align}
where (\ref{conditioning_cannot_increase_5}) is due to the fact
that conditioning cannot increase entropy,
(\ref{product_degraded_imply_7}) and
(\ref{product_degraded_imply_8}) follow from the following Markov
chains
\begin{align}
Z_1^n&\rightarrow Y_{21}^n \rightarrow
\left(W_0,W_1,Y_{22}^{i-1},Y_{12}^n,Z_2^n\right)\\
\left(Y_{12}^{i-1},Z_2^{i-1} \right) & \rightarrow Y_{22}^{i-1}
\rightarrow \left(W_0,W_1,W_2,Y_{21}^n,Z_{2,i}^n,Y_{12,i}\right)
\end{align}
respectively. Thus, plugging (\ref{converse_product_sum_step_15}),
(\ref{memoryless_channel_7}) and
(\ref{converse_product_sum_step_16}) into
(\ref{converse_product_sum_step_1}), we get the following outer
bound on the sum secrecy rate.
\begin{align}
H(W_0,W_1,W_2|Z^n) & \leq \sum_{i=1}^n
I(X_{1,i};Y_{11,i}|Z_{1,i})+I(X_{2,i};Y_{22,i}|U_{2,i},Z_{2,i})+I(U_{2,i};Y_{12,i}|Z_{2,i})+\epsilon_n
\label{product_outer_5}
\end{align}
Following similar steps, we can also get the following one
\begin{align}
H(W_0,W_1,W_2|Z^n) & \leq \sum_{i=1}^n
I(X_{2,i};Y_{22,i}|Z_{2,i})+I(X_{1,i};Y_{11,i}|U_{1,i},Z_{1,i})+I(U_{1,i};Y_{21,i}|Z_{1,i})+\epsilon_n
\label{product_outer_6}
\end{align}
So far, we derived outer bounds, (\ref{product_outer_1}),
(\ref{product_outer_2}), (\ref{product_outer_3}),
(\ref{product_outer_4}), (\ref{product_outer_5}),
(\ref{product_outer_6}), on the capacity region which match the
achievable region provided. The only difference can be on the
joint distribution that they need to satisfy. However, the outer
bounds depend on either $p(u_1,x_1)$ or $p(u_2,x_2)$ but not on
the joint distribution $p(u_1,u_2,x_1,x_2)$. Hence, for the outer
bound, it is sufficient to consider the joint distributions having
the form $p(u_1,u_2,x_1,x_2)=p(u_1,x_1)p(u_2,x_2)$. Thus, the
outer bounds derived and the achievable region coincide yielding
the capacity region.

\subsection{Proof of Theorem~\ref{Theorem_D_PWC_M}}
\label{proof_of_D_PWC_M}

\subsubsection{Achievability}

To show the achievability of the region given in
Theorem~\ref{Theorem_D_PWC_M}, we use
Theorem~\ref{Theorem_product_wiretap_channel}. First, we group
sub-channels into two sets $\mathcal{S}_j, j=1,2,$ where
$\mathcal{S}_j, j=1,2,$ contains the sub-channels in which user
$j$ has the best observation. In other words, we have the Markov
chain
\begin{align}
X_l\rightarrow Y_{1l}\rightarrow Y_{2l}\rightarrow Z_l
\end{align}
for $l\in\mathcal{S}_1,$ and we have this Markov chain
\begin{align}
X_l\rightarrow Y_{2l}\rightarrow Y_{1l}\rightarrow Z_l
\end{align}
for $l\in\mathcal{S}_2$.

We replace $U_j$ with $\{U_l\}_{l\in\mathcal{S}_j}$, $X_j$ with
$\{X_l\}_{l\in\mathcal{S}_j}$, $Y_{j1}$ with
$\{Y_{jl}\}_{l\in\mathcal{S}_1}$, $Y_{j2}$ with
$\{Y_{jl}\}_{l\in\mathcal{S}_2}$, and $Z_{j}$ with
$\{Z_{l}\}_{l\in\mathcal{S}_j}, j=1,2,$ in
Theorem~\ref{Theorem_product_wiretap_channel}. Moreover, if we
select the pairs $\{(U_l,X_l)\}_{l=1}^M$ to be mutually
independent, we get the following joint distribution
\begin{align}
p\left(\{u_l,x_l,y_{1l},y_{2l},z_l\}_{l=1}^M\right)=\prod_{l=1}^M
p(u_l,x_l)p(y_{1l},y_{2l},z_l|x_l)
\end{align}
which implies that random variable tuples
$\{(u_l,x_l,y_{1l},y_{2l},z_l)\}_{l=1}^M$ are mutually
independent. Using this fact, one can reach the expressions given
in Theorem~\ref{Theorem_D_PWC_M}.

\subsubsection{Converse} For the converse part, we again use the
proof of Theorem~\ref{Theorem_product_wiretap_channel}. First,
without loss of generality, we assume
$\mathcal{S}_1=\{1,\ldots,L_1\}$, and
$\mathcal{S}_2=\{L_1+1,\ldots,M\}$. We define the following
auxiliary random variables
\begin{align}
U_{1,i}&= W_0 W_2 Y_{1[L_1+1:M]}^n Y_{1[1:L_1]}^{i-1}
Z_{[1:L_1],i+1}^n \\
U_{2,i}&= W_0 W_1 Y_{2[1:L_1]}^n Y_{2[L_1+1:M]}^{i-1}
Z_{[L_1+1:M],i+1}^n
\end{align}
which satisfy the Markov chains
\begin{align}
U_{1,i}&\rightarrow X_{l,i}\rightarrow
(Y_{1l,i},Y_{2l,i},Z_{l,i}),\quad l=1,\ldots,L_1\\
U_{2,i}&\rightarrow X_{l,i}\rightarrow
(Y_{1l,i},Y_{2l,i},Z_{l,i}),\quad l=L_1+1,\ldots,M
\end{align}
Using the analysis carried out for the proof of
Theorem~\ref{Theorem_product_wiretap_channel}, we get
\begin{align}
nR_0 &\leq \sum_{i=1}^n I(U_{1,i};Y_{1[1:L_1],i}|Z_{[1:L_1],i})+
\sum_{i=1}^n
I(U_{2,i};Y_{1[L_1+1:M],i}|Z_{[L_1+1:M],i})+\epsilon_n
\label{from_prev_theo}
\end{align}
where each term will be treated separately. The first term can be
bounded as follows
\begin{align}
I(U_{1,i};Y_{1[1:L_1],i}|Z_{[1:L_1],i})&=\sum_{l=1}^{L_1}
I(U_{1,i};Y_{1l,i}|Y_{1[1:l-1],i},Z_{[1:L_1],i})\\
&=\sum_{l=1}^{L_1}
I(U_{1,i};Y_{1l,i}|Y_{1[1:l-1],i},Z_{[l:L_1],i})\label{degraded_6}\\
&\leq \sum_{l=1}^{L_1}
I(U_{1,i},Y_{1[1:l-1],i},Z_{[l+1:L_1],i};Y_{1l,i}|Z_{l,i})
\label{dummy_step}
\end{align}
where (\ref{degraded_6}) follows from the Markov chain
\begin{align}
Z_{[1:l-1],i}\rightarrow Y_{1[1:l-1],i}\rightarrow
(U_{1,i},Y_{1l,i},Z_{[l:L_1],i})
\end{align}
which is due to the degradedness of the sub-channels. To this end,
we define the following auxiliary random variables
\begin{align}
V_{l,i}=Y_{1[1:l-1],i}Z_{[l+1:L_1],i}U_{1,i},\quad l=1,\ldots,L_1
\end{align}
which satisfy the Markov chains
\begin{align}
V_{l,i}\rightarrow X_{l,i}\rightarrow
(Y_{1l,i},Y_{2l,i},Z_{l,i}),\quad l=1,\ldots,L_1
\end{align}
Thus, using these new auxiliary random variables in
(\ref{dummy_step}), we get
\begin{align}
I(U_{1,i};Y_{1[1:L_1],i}|Z_{[1:L_1],i})\leq \sum_{l=1}^{L_1}
I(V_{l,i};Y_{1l,i}|Z_{l,i}) \label{intermediate_step_1}
\end{align}
We now bound the second term in (\ref{from_prev_theo}) as follows,
\begin{align}
\lefteqn{I(U_{2,i};Y_{1[L_1+1:M],i}|Z_{[L_1+1:M],i})}\nonumber \\
&=\sum_{l=L_1+1}^M
I(U_{2,i};Y_{1l,i}|Z_{[L_1+1:M],i},Y_{1[L_1+1:l-1],i})\\
&=\sum_{l=L_1+1}^M
I(U_{2,i};Y_{1l,i}|Z_{[l:M],i},Y_{1[L_1+1:l-1],i})
\label{degraded_7} \\
&\leq \sum_{l=L_1+1}^M H(Y_{1l,i}|Z_{l,i})
-H(Y_{1l,i}|Z_{[l:M],i},Y_{1[L_1+1:l-1],i},U_{2,i}) \label{conditioning_cannot_increase_8}\\
&\leq \sum_{l=L_1+1}^M H(Y_{1l,i}|Z_{l,i})
-H(Y_{1l,i}|Z_{[l:M],i},Y_{1[L_1+1:l-1],i},U_{2,i},Y_{2[L_1+1:l-1],i})\label{conditioning_cannot_increase_9}\\
&= \sum_{l=L_1+1}^M H(Y_{1l,i}|Z_{l,i})
-H(Y_{1l,i}|Z_{[l:M],i},U_{2,i},Y_{2[L_1+1:l-1],i}) \label{degraded_8} \\
&= \sum_{l=L_1+1}^M
I(Z_{[l+1:M],i},U_{2,i},Y_{2[L_1+1:l-1],i};Y_{1l,i}|Z_{l,i})
\label{dummy_step_1}
\end{align}
where (\ref{degraded_7}) follows from the Markov chain
\begin{align}
Z_{[L_1+1:l-1],i}\rightarrow Y_{1[L_1+1:l-1],i} \rightarrow
(U_{2,i},Z_{[l:M],i},Y_{1l,i})
\end{align}
which is a consequence of the degradedness of the sub-channels,
(\ref{conditioning_cannot_increase_8}) and
(\ref{conditioning_cannot_increase_9}) follow from the fact that
conditioning cannot increase entropy, and (\ref{degraded_8}) is
due to the Markov chain
\begin{align}
Y_{1[L_1+1:l-1],i}\rightarrow Y_{2[L_1+1:l-1],i} \rightarrow
(U_{2,i},Z_{[l:M],i},Y_{1l,i})
\end{align}
which is again a consequence of the degradedness of the
sub-channels. To this end, we define the following auxiliary
random variables
\begin{align}
V_{l,i}=Y_{2[L_1+1:l-1],i} Z_{[l+1:M],i}U_{2,i},\quad
l=L_1+1,\ldots,M \label{def_v_2}
\end{align}
which satisfy the Markov chains
\begin{align}
V_{l,i}\rightarrow X_{l,i}\rightarrow
(Y_{1l,i},Y_{2l,i},Z_{l,i}),\quad l=L_1+1,\ldots,M
\end{align}
Thus, using these new auxiliary random variables in
(\ref{dummy_step_1}), we get
\begin{align}
I(U_{2,i};Y_{1[L_1+1:M],i}|Z_{[L_1+1:M],i})\leq \sum_{l=L_1+1}^M
I(V_{l,i};Y_{1l,i}|Z_{l,i}) \label{intermediate_step_2}
\end{align}
Finally, using (\ref{intermediate_step_1}) and
(\ref{intermediate_step_2}) in (\ref{from_prev_theo}), we obtain
\begin{align}
nR_0 &\leq \sum_{i=1}^n \sum_{l=1}^M
I(V_{l,i};Y_{1l,i}|Z_{l,i})+\epsilon_n
\end{align}
Due to symmetry, we also have
\begin{align}
nR_0 &\leq \sum_{i=1}^n \sum_{l=1}^M
I(V_{l,i};Y_{2l,i}|Z_{l,i})+\epsilon_n
\end{align}

We now bound the sum of common and independent message rates.
Using the converse proof of
Theorem~\ref{Theorem_product_wiretap_channel}, we get
\begin{align}
n(R_0+R_1)&\leq \sum_{i=1}^n
I(X_{[1:L_1],i};Y_{1[1:L_1],i}|Z_{[1:L_1],i})+ \sum_{i=1}^n
I(U_{2,i};Y_{1[L_1+1:M],i}|Z_{[L_1+M],i})+\epsilon_n
\label{from_prev_theo_1}
\end{align}
where, for the second term we already obtained an outer bound
given in (\ref{intermediate_step_2}). We now bound the first term,
\begin{align}
I(X_{[1:L_1],i};Y_{1[1:L_1],i}|Z_{[1:L_1],i})&=\sum_{l=1}^{L_1}
I(X_{[1:L_1],i};Y_{1l,i}|Z_{[1:L_1],i},Y_{1[1:l-1],i}) \\
&\leq \sum_{l=1}^{L_1} H(Y_{1l,i}|Z_{l,i})
-H(Y_{1l,i}|Z_{[1:L_1],i},Y_{1[1:l-1],i},X_{[1:L_1],i})
\label{conditioning_cannot_increase_10}\\
&= \sum_{l=1}^{L_1} H(Y_{1l,i}|Z_{l,i})
-H(Y_{1l,i}|Z_{l,i},X_{l,i}) \label{memoryless_channel_12}\\
&= \sum_{l=1}^{L_1} I(X_{l,i};Y_{1l,i}|Z_{l,i})
\label{intermediate_step_3}
\end{align}
where (\ref{conditioning_cannot_increase_10}) follows from the
fact that conditioning cannot increase entropy, and
(\ref{memoryless_channel_12}) is due to the following Markov chain
\begin{align}
(Y_{1l,i},Z_{l,i})\rightarrow X_{l,i}\rightarrow
(X_{[1:l-1],i},X_{[l+1:L_1],i},Y_{1[1:l-1],i}Z_{[1:l-1],i},Z_{[l+1:L_1],i})
\end{align}
which follows from the facts that channel is memoryless and
sub-channels are independent. Thus, plugging
(\ref{intermediate_step_2}) and (\ref{intermediate_step_3}) into
(\ref{from_prev_theo_1}), we obtain
\begin{align}
n(R_0+R_1)&\leq
\sum_{i=1}^n\sum_{l\in\mathcal{S}_1}I(X_{l,i};Y_{1l,i}|Z_{l,i})+
\sum_{i=1}^n\sum_{l\in\mathcal{S}_2}I(V_{l,i};Y_{1l,i}|Z_{l,i})+\epsilon_n
\end{align}
Due to symmetry, we also have
\begin{align}
n(R_0+R_1)&\leq
\sum_{i=1}^n\sum_{l\in\mathcal{S}_2}I(X_{l,i};Y_{2l,i}|Z_{l,i})+
\sum_{i=1}^n\sum_{l\in\mathcal{S}_1}I(V_{l,i};Y_{2l,i}|Z_{l,i})+\epsilon_n
\end{align}

We now bound the sum secrecy rate. We first borrow the following
outer bound from the converse proof of
Theorem~\ref{Theorem_product_wiretap_channel},
\begin{align}
\lefteqn{n(R_0+R_1+R_2)\leq \sum_{i=1}^n
I(X_{[1:L_1],i};Y_{1[1:L_1],i}|Z_{[1:L_1],i})}\\
&+ \sum_{i=1}^n
I(X_{[L_1+1:M],i};Y_{2[L_1+1:M],i}|U_{2,i},Z_{[L_1+1:M],i})
+\sum_{i=1}^n I(U_{2,i};Y_{1[L_1+1:M],i}|Z_{[L_1+1:M],i})
\label{from_prev_theo_2}
\end{align}
where, for the first and third terms we already obtained outer
bounds given in (\ref{intermediate_step_3}) and
(\ref{intermediate_step_2}), respectively. We now bound the second
term as follows,
\begin{align}
\lefteqn{I(X_{[L_1+1:M],i};Y_{2[L_1+1:M],i}|U_{2,i},Z_{[L_1+1:M],i})}\hspace{3cm}\nonumber\\
&=\sum_{l=L_1+1}^M
I(X_{[L_1+1:M],i};Y_{2l,i}|U_{2,i},Z_{[L_1+1:M],i},Y_{2[L_1+1:l-1],i})\\
&=\sum_{l=L_1+1}^M
I(X_{[L_1+1:M],i};Y_{2l,i}|U_{2,i},Z_{[l:M],i},Y_{2[L_1+1:l-1],i})\label{degraded_9}\\
&=\sum_{l=L_1+1}^M
I(X_{[L_1+1:M],i};Y_{2l,i}|V_{l,i},Z_{l,i})\label{use_def_v_2}\\
&=\sum_{l=L_1+1}^M H(Y_{2l,i}|V_{l,i},Z_{l,i})
-H(Y_{2l,i}|V_{l,i},Z_{l,i},X_{[L_1+1:M],i})\\
&=\sum_{l=L_1+1}^M H(Y_{2l,i}|V_{l,i},Z_{l,i})
-H(Y_{2l,i}|V_{l,i},Z_{l,i},X_{l,i})\label{memoryless_channel_13}\\
&=\sum_{l=L_1+1}^M I(X_{l,i};Y_{2l,i}|V_{l,i},Z_{l,i})
\label{intermediate_step_4}
\end{align}
where (\ref{degraded_9}) follows from the Markov chain
\begin{align}
Z_{[L_1+1:l-1],i}\rightarrow Y_{2[L_1+1:l-1],i}\rightarrow
U_{2,i}, Z_{[l:M],i},X_{[L_1+1:M],i},Y_{2l,i}
\end{align}
which is a consequence of the degradedness of the sub-channels,
(\ref{use_def_v_2}) is obtained via using the definition of
$V_{2,i}$ given in (\ref{def_v_2}), and
(\ref{memoryless_channel_13}) follows from the Markov chain
\begin{align}
(Z_{l,i},Y_{2l,i})\rightarrow X_{l,i}\rightarrow (V_{l,i},
X_{[L_1+1:l-1],i},X_{[l+1:M]})
\end{align}
which is due to the facts that channel is memoryless and
sub-channels are independent. Thus, plugging
(\ref{intermediate_step_2}), (\ref{intermediate_step_3}) and
(\ref{intermediate_step_4}) into (\ref{from_prev_theo_2}), we get
\begin{align}
n(R_0+R_1+R_2)&\leq \sum_{i=1}^n
\sum_{l\in\mathcal{S}_1}I(X_{l,i};Y_{1l,i}|Z_{l,i})+ \sum_{i=1}^n
\sum_{l\in\mathcal{S}_2}I(X_{l,i};Y_{2l,i}|V_{l,i},Z_{l,i}) \nonumber \\
&\quad + \sum_{i=1}^n \sum_{l\in\mathcal{S}_2}
I(V_{l,i};Y_{1l,i}|Z_{l,i}) +\epsilon_n
\end{align}
Due to symmetry, we also have
\begin{align}
n(R_0+R_1+R_2) &\leq \sum_{i=1}^n
\sum_{l\in\mathcal{S}_2}I(X_{l,i};Y_{2l,i}|Z_{l,i})+ \sum_{i=1}^n
\sum_{l\in\mathcal{S}_1}I(X_{l,i};Y_{1l,i}|V_{l,i},Z_{l,i})
\nonumber\\
&\quad+ \sum_{i=1}^n \sum_{l\in\mathcal{S}_1}
I(V_{l,i};Y_{2l,i}|Z_{l,i}) +\epsilon_n
\end{align}
Finally, we note that all outer bounds depend on the distributions
$p(v_{l,i},x_{l,i},y_{1l,i},y_{2l,i},z_{l,i})=p(v_{l,i},x_{l,i})p(y_{1l,i},y_{2l,i},z_{l,i}|x_{l,i})$
but not on any joint distributions of the tuples
$(v_{l,i},x_{l,i},y_{1l,i},y_{2l,i},\break z_{l,i})$ implying that
selection of the pairs $(v_{l,i},x_{l,i})$ to be mutually
independent is optimum.

\section{Proof of Theorem~\ref{Theorem_Sum_Unmatched}}
\label{Proof_Sum_Unmatched}

We prove Theorem~\ref{Theorem_Sum_Unmatched} in two parts; first,
we show achievability, and then we prove the converse.

\subsection{Achievability}
Similar to what we have done to show the achievability of
Theorem~\ref{Theorem_product_wiretap_channel}, we first note that
boundary of the capacity region can be decomposed into three
surfaces \cite{Product_Broadcast}.

\begin{itemize}
\item First surface:
\begin{align}
R_0&\leq \bar{\alpha}I(U_2;Y_{12}|Z_2)\\
R_2&\leq \bar{\alpha}I(X_2;Y_{22}|U_2,Z_2)\\
R_0+R_1&\leq \alpha
I(X_1;Y_{11}|Z_1)+\bar{\alpha}I(U_2;Y_{12}|Z_2),\quad U_1=\phi
\end{align}
\item Second surface:
\begin{align}
R_0&\leq \alpha I(U_1;Y_{21}|Z_1)\\
R_1&\leq \alpha I(X_1;Y_{11}|U_1,Z_1)\\
R_0+R_2 &\leq \alpha
I(U_1;Y_{21}|Z_1)+\bar{\alpha}I(X_2;Y_{22}|Z_2),\quad U_2=\phi
\end{align}
\item Third surface:
\begin{align}
R_1&\leq\alpha I(X_1;Y_{11}|U_1,Z_1)\\
R_2&\leq \bar{\alpha}I(X_2;Y_{22}|U_2,Z_2)\\
R_0&\leq \alpha I(U_1;Y_{11}|Z_1)+\bar{\alpha}I(U_2;Y_{12}|Z_2)\\
R_0&\leq \alpha I(U_1;Y_{21}|Z_1)+\bar{\alpha}I(U_2;Y_{22}|Z_2)
\end{align}
\end{itemize}
To show the achievability of each surface, we first introduce a
codebook structure.

\vspace{0.25cm} \noindent \underline{\textbf{Codebook structure:}}

\vspace{0.15cm} Fix the probability distribution as,
\begin{align}
p(u_1,x_1)p(u_2,x_2)p(y_1,y_2,z|x)
\end{align}
\begin{itemize}
\item Generate $2^{n(R_{01}+R_{11}+\tilde{R}_{11})}$ length-$n_1$
sequences $\bold{u}_1$ through
$p(\bold{u}_1)=\prod_{i=1}^{n_1}p(u_{1,i})$ and index them as
$\bold{u}_1 (w_{01},w_{11},\tilde{w}_{11})$ where
$w_{01}\in\{1,\ldots,2^{nR_{01}}\}$,
$w_{11}\in\{1,\ldots,2^{nR_{11}}\}$ and
$\tilde{w}_{11}\in\{1,\ldots,2^{n\tilde{R}_{11}}\}$.

\item For each $\bold{u}_1$, generate
$2^{n(R_{12}+\tilde{R}_{12})}$ length-$n_1$ sequences $\bold{x}_1$
through $p(\bold{x}_1)=\break\prod_{i=1}^{n_1}p(x_{1,i}|u_{1,i})$
and index them as $\bold{x}_1
(w_{01},w_{11},\tilde{w}_{11},w_{12},\tilde{w}_{12})$ where
$w_{12}\in\{1,\ldots,2^{nR_{12}}\}$,
$\tilde{w}_{12}\in\{1,\ldots,2^{n\tilde{R}_{12}}\}$.

\item Generate $2^{n(R_{02}+R_{21}+\tilde{R}_{21})}$
length-$(n-n_1)$ sequences $\bold{u}_2$ through
$p(\bold{u}_2)=\prod_{i=1}^{n-n_1}p(u_{2,i})$ and index them as
$\bold{u}_2 (w_{02},w_{21},\tilde{w}_{21})$ where
$w_{02}\in\{1,\ldots,2^{nR_{02}}\}$,
$w_{21}\in\{1,\ldots,2^{nR_{21}}\}$ and
$\tilde{w}_{21}\in\{1,\ldots,2^{n\tilde{R}_{21}}\}$.

\item For each $\bold{u}_2$, generate
$2^{n(R_{22}+\tilde{R}_{22})}$ length-$(n-n_1)$ sequences
$\bold{x}_2$ through
$p(\bold{x}_2)=\break\prod_{i=1}^{n-n_1}p(x_{2,i}|u_{2,i})$ and
index them as $\bold{x}_2
(w_{02},w_{21},\tilde{w}_{21},w_{22},\tilde{w}_{22})$ where
$w_{22}\in\{1,\ldots,\break 2^{nR_{22}}\}$,
$\tilde{w}_{22}\in\{1,\ldots,2^{n\tilde{R}_{22}}\}$.

\item We remark that this codebook uses first channel $n_1$ times
and the other one $(n-n_1)$ times. We define
\begin{align}
\alpha=\frac{n_1}{n}
\end{align}
and $\bar{\alpha}=1-\alpha$.

\item Furthermore, we set
\begin{align}
\tilde{R}_{11}&=\alpha I(U_1;Z_1)\label{sum_dummy_defns_1}\\
\tilde{R}_{12}&=\alpha I(X_1;Z_1|U_1)\label{sum_dummy_defns_2}\\
\tilde{R}_{21}&=\bar{\alpha} I(U_2;Z_2)\label{sum_dummy_defns_3}\\
\tilde{R}_{22}&=\bar{\alpha} I(X_2;Z_2|U_2)\label{sum_dummy_defns_4}\\
R_1&=R_{11}+R_{12}\label{sum_dummy_defns_5}\\
R_2&=R_{21}+R_{22}\label{sum_dummy_defns_6}
\end{align}

\end{itemize}

\vspace{0.25cm} \noindent \underline{\textbf{Encoding:}}

\vspace{0.15cm}When the transmitted messages are
$(w_{01},w_{02},w_{11},w_{12},w_{21},w_{22})$, we randomly pick
$(\tilde{w}_{11},\break\tilde{w}_{12},\tilde{w}_{21},\tilde{w}_{22})$
and send corresponding codewords.

\vspace{0.25cm} \noindent \underline{\textbf{Decoding:}}

\vspace{0.15cm}

Using this codebook structure, we can show that all three surfaces
which determine the boundary of the capacity region are
achievable. For example, if we set $U_1=\phi$ (that implies
$R_{01}=R_{11}=\tilde{R}_{11}=0$) and $R_{21}=0$, then we achieve
the following rates with vanishingly small error probability.
\begin{align}
R_1&\leq \alpha I(X_1;Y_{11}|Z_1)\\
R_0&\leq \bar{\alpha} I(U_2;Y_{12}|Z_2)\\
R_2&\leq \bar{\alpha} I(X_2;Y_{22}|U_2,Z_2)
\end{align}
Exchanging common message rate with user 1's independent message
rate, one can obtain the first surface. Second surface follows
from symmetry. For the third surface, we first set
$R_{11}=R_{21}=0$. Moreover, we send common message in its
entirety, i.e., we do not use a rate splitting for the common
message, hence we set $R_{01}=R_{02}=R_0$, $w_{01}=w_{02}=w_0$. In
this case, each user, say the $j$th one, decodes the common
message by looking for a unique $w_0$ which satisfies
\begin{align}
E_{j1}^{w_0}=\left\{\exists \tilde{w}_{01}:
(\bold{u}_1(w_{0},\tilde{w}_{01}),\bold{y}_{j1})\in
A_{\epsilon}^n\right\}\\
E_{j2}^{w_0}=\left\{\exists \tilde{w}_{02}:
(\bold{u}_2(w_{0},\tilde{w}_{02}),\bold{y}_{j2})\in
A_{\epsilon}^n\right\}\\
\end{align}
Following the analysis carried out in
(\ref{sum_unmatch_third_surface_start})-(\ref{third_surface_12_derive}),
the sufficient conditions for the common message to be decodable
by both users can be found as
\begin{align}
R_0 &\leq \alpha I(U_1;Y_{j1}|Z_1)+\bar{\alpha} I(U_2;Y_{j2}|Z_2),
\quad j=1,2
\end{align}
After decoding the common message, each user can decode its
independent message if
\begin{align}
R_1&\leq \alpha I(X_1;Y_{11}|U_1,Z_1)\\
R_2&\leq \bar{\alpha} I(X_2;Y_{22}|U_2,Z_2)
\end{align}
Thus, the third surface can be achieved with vanishingly small
error probability. As of now, we showed that all rates in the
so-called capacity region are achievable with vanishingly small
error probability, however we did not claim anything about the
secrecy conditions which will be considered next.

\vspace{0.25cm} \noindent \underline{\textbf{Equivocation
calculation:}}

\vspace{0.15cm} To complete the achievability part of the proof,
we need to show that this codebook structure also satisfies the
secrecy conditions. For that purpose, it is sufficient to consider
the sum rate secrecy condition.
\begin{align}
\lefteqn{H(W_0,W_1,W_2|Z_1^{n_1},Z_2^{n-n_1})=H(W_0,W_1,W_2,Z_1^{n_1},Z_2^{n-n_1})-H(Z_1^{n_1},Z_2^{n-n_1})}&\\
&=H(W_0,W_1,W_2,U_1^{n_1},U_2^{n-n_1},X_1^{n_1},X_2^{n-n_1},Z_1^{n_1},Z_2^{n-n_1})
-H(Z_1^{n_1},Z_2^{n-n_1})\nonumber\\
&\quad
-H(U_1^{n_1},U_2^{n-n_1},X_1^{n_1},X_2^{n-n_1}|W_0,W_1,W_2,Z_1^{n_1},Z_2^{n-n_1})\\
&=H(U_1^{n_1},U_2^{n-n_1},X_1^{n_1},X_2^{n-n_1})
+H(W_0,W_1,W_2,Z_1^{n_1},Z_2^{n-n_1}|U_1^{n_1},U_2^{n-n_1},X_1^{n_1},X_2^{n-n_1})\nonumber\\
&\quad -H(Z_1^{n_1},Z_2^{n-n_1})
-H(U_1^{n_1},U_2^{n-n_1},X_1^{n_1},X_2^{n-n_1}|W_0,W_1,W_2,Z_1^{n_1},Z_2^{n-n_1})\\
&\geq H(U_1^{n_1},U_2^{n-n_1},X_1^{n_1},X_2^{n-n_1})
+H(Z_1^{n_1},Z_2^{n-n_1}|U_1^{n_1},U_2^{n-n_1},X_1^{n_1},X_2^{n-n_1})\nonumber\\
&\quad -H(Z_1^{n_1},Z_2^{n-n_1})
-H(U_1^{n_1},U_2^{n-n_1},X_1^{n_1},X_2^{n-n_1}|W_0,W_1,W_2,Z_1^{n_1},Z_2^{n-n_1})
\label{equi_comp_sum_step1}
\end{align}
where each term will be treated separately. The first term is
\begin{align}
\lefteqn{H(U_1^{n_1},U_2^{n-n_1},X_1^{n_1},X_2^{n-n_1})}&\nonumber \\
&=H(U_1^{n_1},U_2^{n-n_1})+H(X_1^{n_1}|U_1^{n_1})+H(X_2^{n-n_1}|U_2^{n-n_1})\\
&=n(R_0+R_{11}+\tilde{R}_{11}+R_{21}+\tilde{R}_{21})+n(R_{12}+\tilde{R}_{12})+n(R_{22}+\tilde{R}_{22})\label{equi_comp_sum_step2}\\
&=n(R_0+R_1+R_2)+n_1 I(X_1;Z_1)+(n-n_1)I(X_2;Z_2)
\label{equi_comp_sum_step3}
\end{align}
where the first equality is due to the Markov chain
\begin{align}
X_1^{n_1}\rightarrow U_1^{n_1}\rightarrow U_2^{n-n_1}\rightarrow
X_2^{n-n_1}
\end{align}
The equality in (\ref{equi_comp_sum_step2}) is due to the fact
that $(U_1^{n_1},U_2^{n-n_1})$ can take
$2^{n(R_0+R_{11}+\tilde{R}_{11}+R_{21}+\tilde{R}_{21})}$ values
uniformly, and given $U_1^{n_1}$ (resp. $U_2^{n-n_1}$),
$X_1^{n_1}$ (resp. $X_2^{n-n_1}$) can take
$2^{n(R_{12}+\tilde{R}_{12})}$ (resp.
$2^{n(R_{22}+\tilde{R}_{22})}$) values with equal probability. To
reach (\ref{equi_comp_sum_step3}), we use the definitions in
(\ref{sum_dummy_defns_1})-(\ref{sum_dummy_defns_6}). We consider
the second and third terms in (\ref{equi_comp_sum_step1}).
\begin{align}
\lefteqn{H(Z_1^{n_1},Z_2^{n-n_1})-H(Z_1^{n_1},Z_2^{n-n_1}|U_1^{n_1},U_2^{n-n_1},X_1^{n_1},X_2^{n-n_1})}&
\nonumber\\
&\leq
H(Z_1^{n_1})+H(Z_2^{n-n_1})-H(Z_1^{n_1},Z_2^{n-n_1}|U_1^{n_1},U_2^{n-n_1},X_1^{n_1},X_2^{n-n_1})\label{conditioning_cannot_increase_7}\\
&=
H(Z_1^{n_1})+H(Z_2^{n-n_1})-H(Z_1^{n_1}|X_1^{n_1})+H(Z_2^{n-n_1}|X_2^{n-n_1})\label{memoryless_channel_9}\\
&= I(X_1^{n_1};Z_1^{n_1})+I(X_2^{n-n_1};Z_2^{n-n_1})\\
&\leq n_1 I(X_1;Z_1)+(n-n_1)I(X_2;Z_2)+\gamma_{1,n}+\gamma_{2,n}
\label{memoryless_channel_10}
\end{align}
where (\ref{conditioning_cannot_increase_7}) is due to the fact
that conditioning cannot increase entropy,
(\ref{memoryless_channel_9}) follows from the Markov chain
\begin{align}
Z_1^{n_1}\rightarrow X_1^{n_1}\rightarrow U_1^{n_1}\rightarrow
U_2^{n-n_1}\rightarrow X_2^{n-n_1}\rightarrow Z_2^{n-n_1}
\end{align}
and (\ref{memoryless_channel_10}) can be shown using the technique
devised in~\cite{Wyner}. We bound the fourth term of
(\ref{equi_comp_sum_step1}). To this end, assume that the
eavesdropper tries to decode
$(U_1^{n_1},X_1^{n_1},U_2^{n-n_1},X_2^{n-n_1})$ given side
information $(W_0=w_0,W_1=w_1,W_2=w_2)$. Since the confusion
message rates are selected to ensure that (see
(\ref{sum_dummy_defns_1})-(\ref{sum_dummy_defns_4})) the
eavesdropper can decode them as long as side information is
available. Consequently, use of Fano's lemma yields
\begin{align}
H(U_1^{n_1},U_2^{n-n_1},X_1^{n_1},X_2^{n-n_1}|W_0,W_1,W_2,Z_1^{n_1},Z_2^{n-n_1})<\epsilon_n
\label{eavesdropper_can_decode}
\end{align}
Finally, plugging
(\ref{equi_comp_sum_step3}),(\ref{memoryless_channel_10}) and
(\ref{eavesdropper_can_decode}) into (\ref{equi_comp_sum_step1}),
we get
\begin{align}
H(W_0,W_1,W_2|Z_1^{n_1},Z_2^{n-n_1})\geq
n(R_0+R_1+R_2)-\epsilon_n-\gamma_{1,n}-\gamma_{2,n}
\end{align}
which completes the achievability part of the proof.

\subsection{Converse}

First, let us define the following auxiliary random variables,
\begin{align}
U_{1,i}&=W_0 W_2 Y_{12}^{n-n_1}Y_{11}^{i-1}Z_{1,i+1}^{n_1},\quad
i=1,\ldots,n_1 \label{aux_sum_1}\\
U_{2,i}&= W_0W_1 Y_{21}^{n_1}Y_{22}^{i-1}Z_{2,i+1}^{n-n_1},\qquad
i=1,\ldots,n-n_1  \label{aux_sum_2}
\end{align}
where we assume that first channel is used $n_1$ times. We again
define
\begin{align}
\alpha=\frac{n_1}{n}
\end{align}
We note that auxiliary random variables, $U_{1,i},U_{2,i}$ satisfy
the Markov chains
\begin{align}
U_{1,i}\rightarrow X_{1,i}\rightarrow
(Y_{11,i},Y_{21,i},Z_{1,i})\\
U_{2,i}\rightarrow X_{2,i}\rightarrow (Y_{21,i},Y_{22,i},Z_{2,i})
\end{align}
Similar to the converse of
Theorem~\ref{Theorem_product_wiretap_channel}, here again,
$U_{1,i}$ and $U_{2,i}$ can be arbitrarily correlated. However, at
the end of converse, it will be clear that selection of them as
independent would yield the same region. Start with the common
message rate,
\begin{align}
\lefteqn{H(W_0|Z_1^{n_1},Z_2^{n-n_1})}&\\
&\leq
I(W_0;Y_{11}^{n_1},Y_{12}^{n-n_1})-I(W_0;Z_{1}^{n_1},Z_{2}^{n-n_1})+\epsilon_n \label{Fano_7}\\
&=I(W_0;Y_{11}^{n_1},Y_{12}^{n-n_1}|Z_{1}^{n_1},Z_{2}^{n-n_1})+\epsilon_n
\label{degraded_3} \\
&=I(W_0;Y_{12}^{n-n_1}|Z_{1}^{n_1},Z_{2}^{n-n_1})+I(W_0;Y_{11}^{n_1}|Z_{1}^{n_1},Z_{2}^{n-n_1},Y_{12}^{n-n_1})+\epsilon_n
\\
&\leq
I(W_0,W_1;Y_{12}^{n-n_1}|Z_{1}^{n_1},Z_{2}^{n-n_1})+I(W_0,W_2;Y_{11}^{n_1}|Z_{1}^{n_1},Z_{2}^{n-n_1},Y_{12}^{n-n_1})+\epsilon_n
\label{converse_common_sum_step1}
\end{align}
where (\ref{Fano_7}) is due to Fano's lemma, (\ref{degraded_3}) is
due to the fact that the eavesdropper's channel is degraded with
respect to the first user's channel. Once we obtain
(\ref{converse_common_sum_step1}), using the analysis carried out
in the proof of Theorem~\ref{Theorem_product_wiretap_channel}, we
can obtain the following bounds.
\begin{align}
I(W_0,W_1;Y_{12}^{n-n_1}|Z_{1}^{n_1},Z_{2}^{n-n_1})&\leq
\sum_{i=1}^{n-n_1}I(U_{2,i};Y_{12,i}|Z_{2,i})
\label{converse_common_sum_step2}\\
I(W_0,W_2;Y_{11}^{n_1}|Z_{1}^{n_1},Z_{2}^{n-n_1},Y_{12}^{n-n_1})
&\leq \sum_{i=1}^{n_1}I(U_{1,i};Y_{11,i}|Z_{1,i})
\label{converse_common_sum_step3}
\end{align}
where (\ref{converse_common_sum_step2}) (resp.
(\ref{converse_common_sum_step3})) can be derived following the
lines from (\ref{converse_common_use_1_start}) (resp.
(\ref{converse_common_use_2_start})) to
(\ref{converse_product_common_step_11}) (resp.
(\ref{converse_product_common_step_12})). Thus, we have
\begin{align}
H(W_0|Z_1^{n_1},Z_2^{n-n_1})&\leq
\sum_{i=1}^{n-n_1}I(U_{2,i};Y_{12,i}|Z_{2,i}) +
 \sum_{i=1}^{n_1}I(U_{1,i};Y_{11,i}|Z_{1,i})+\epsilon_n
\end{align}
and similarly, we can get
\begin{align}
H(W_0|Z_1^{n_1},Z_2^{n-n_1})&\leq
\sum_{i=1}^{n-n_1}I(U_{2,i};Y_{22,i}|Z_{2,i}) +
 \sum_{i=1}^{n_1}I(U_{1,i};Y_{21,i}|Z_{1,i})+\epsilon_n
\end{align}
We now consider the sum of common and independent message rates,
\begin{align}
\lefteqn{H(W_0,W_1|Z_1^{n_1},Z_2^{n-n_1})}&\nonumber \\
&\leq
I(W_0,W_1;Y_{11}^{n_1},Y_{12}^{n-n_1})-I(W_0,W_1;Z_1^{n_1},Z_2^{n-n_1})+\epsilon_n
\label{Fano_8}\\
&=I(W_0,W_1;Y_{11}^{n_1},Y_{12}^{n-n_1}|Z_1^{n_1},Z_2^{n-n_1})+\epsilon_n
\label{degraded_4}\\
&=I(W_0,W_1;Y_{12}^{n-n_1}|Z_1^{n_1},Z_2^{n-n_1})+I(W_0,W_1;Y_{11}^{n_1}|Z_1^{n_1},Z_2^{n-n_1},Y_{12}^{n-n_1})+\epsilon_n
\label{converse_common_private_sum_step1}
\end{align}
where (\ref{Fano_8}) is due to Fano's lemma, (\ref{degraded_4})
follows from the fact that the eavesdropper's channel is degraded
with respect to the first user's channel. The first term of
(\ref{converse_common_private_sum_step1}) is already bounded in
(\ref{converse_common_sum_step2}). The second term can be bounded
as
\begin{align}
I(W_0,W_1;Y_{11}^{n_1}|Z_1^{n_1},Z_2^{n-n_1},Y_{12}^{n-n_1})&\leq
\sum_{i=1}^{n_1}I(X_{1,i};Y_{11,i}|Z_{1,i})
\label{converse_common_private_sum_step2}
\end{align}
which can be obtained following the lines from
(\ref{converse_common_use_3_start}) to
(\ref{converse_product_common_private_step_12}). Hence, plugging
(\ref{converse_common_sum_step2}) and
(\ref{converse_common_private_sum_step2}) into
(\ref{converse_common_private_sum_step1}), we get
\begin{align}
H(W_0,W_1|Z_1^{n_1},Z_2^{n-n_1})\leq
\sum_{i=1}^{n-n_1}I(U_{2,i};Y_{12,i}|Z_{2,i})+
\sum_{i=1}^{n_1}I(X_{1,i};Y_{11,i}|Z_{1,i})+\epsilon_n
\end{align}
Similarly, we can obtain
\begin{align}
H(W_0,W_2|Z_1^{n_1},Z_2^{n-n_1})\leq
\sum_{i=1}^{n-n_1}I(X_{2,i};Y_{22,i}|Z_{2,i})+
\sum_{i=1}^{n_1}I(U_{1,i};Y_{21,i}|Z_{1,i})+\epsilon_n
\end{align}
Finally, we derive the outer bounds for the sum secrecy rate,
\begin{align}
\lefteqn{H(W_0,W_1,W_2|Z_{1}^{n_1},Z_2^{n-n_1})\leq
I(W_0,W_1;Y_{11}^{n_1},Y_{12}^{n-n_1})+I(W_2;Y_{21}^{n_1},Y_{22}^{n-n_1}|W_0,W_1)}&\nonumber\\
&\quad -I(W_0,W_1,W_2;Z_1^{n_1},Z_2^{n-n_1})+\epsilon_n
\label{Fano_9}
\\
&=I(W_0,W_1;Y_{11}^{n_1},Y_{12}^{n-n_1}|Z_1^{n_1},Z_2^{n-n_1})+I(W_2;Y_{21}^{n_1},Y_{22}^{n-n_1}|W_0,W_1,Z_1^{n_1},Z_2^{n-n_1})+\epsilon_n
\label{degraded_5}\\
&= I(W_0,W_1;Y_{12}^{n-n_1}|Z_1^{n_1},Z_2^{n-n_1})
+I(W_0,W_1;Y_{11}^{n_1}|Z_1^{n_1},Z_2^{n-n_1},Y_{12}^{n-n_1})\nonumber
\\
&\quad +I(W_2;Y_{21}^{n_1}|W_0,W_1,Z_1^{n_1},Z_2^{n-n_1})
+I(W_2;Y_{22}^{n-n_1}|W_0,W_1,Z_1^{n_1},Z_2^{n-n_1},Y_{21}^{n_1})
+\epsilon_n\\
&=
I(W_0,W_1,Y_{21}^{n_1};Y_{12}^{n-n_1}|Z_1^{n_1},Z_2^{n-n_1})-I(Y_{21}^{n_1};Y_{12}^{n-n_1}|Z_1^{n_1},Z_2^{n-n_1},W_0,W_1)\nonumber\\
&\quad
+I(W_0,W_1;Y_{11}^{n_1}|Z_1^{n_1},Z_2^{n-n_1},Y_{12}^{n-n_1})
+I(W_2;Y_{21}^{n_1}|W_0,W_1,Z_1^{n_1},Z_2^{n-n_1}) \nonumber \\
&\quad
+I(W_2;Y_{22}^{n-n_1}|W_0,W_1,Z_1^{n_1},Z_2^{n-n_1},Y_{21}^{n_1})
+\epsilon_n\\
&=S_1-S_2+S_3+S_4+S_5+\epsilon_n
\label{converse_common_sum_rate_sum_step1}
\end{align}
where in (\ref{Fano_9}), we used Fano's lemma and
(\ref{degraded_5}) follows from the fact that the eavesdropper's
channel is degraded with respect to both users' channels. We can
again use the analysis carried out in the converse proof of
Theorem~\ref{Theorem_product_wiretap_channel} to bound
(\ref{converse_common_sum_rate_sum_step1}). For example, following
lines from (\ref{converse_common_use_4_start}) to
(\ref{converse_product_sum_step_15}), we can obtain
\begin{align}
S_4+S_3-S_2 \leq \sum_{i=1}^{n_1}I(X_{1,i};Y_{11,i}|Z_{1,i})
\label{converse_common_sum_rate_sum_step2}
\end{align}
Similarly, if we follow the analysis from
(\ref{converse_common_use_5_start}) to
(\ref{memoryless_channel_7}), we can get
\begin{align}
S_5\leq \sum_{i=1}^{n-n_1}I(X_{2,i};Y_{22,i}|U_{2,i},Z_{2,i})
\label{converse_common_sum_rate_sum_step3}
\end{align}
and if we follow the lines from
(\ref{converse_common_use_6_start}) to
(\ref{converse_product_sum_step_16}), we can get
\begin{align}
S_1&\leq \sum_{i=1}^{n-n_1}I(U_{2,i};Y_{12,i}|Z_{2,i})
\label{converse_common_sum_rate_sum_step4}
\end{align}
Thus, plugging (\ref{converse_common_sum_rate_sum_step2}),
(\ref{converse_common_sum_rate_sum_step3}) and
(\ref{converse_common_sum_rate_sum_step4}) into
(\ref{converse_common_sum_rate_sum_step1}), we get
\begin{align}
H(W_0,W_1,W_2|Z_1^{n_1},Z_2^{n-n_1})&\leq
\sum_{i=1}^{n_1}I(X_{1,i};Y_{11,i}|Z_{1,i})
+\sum_{i=1}^{n-n_1}I(U_{2,i};Y_{12,i}|Z_{2,i})\nonumber\\
&\quad
+\sum_{i=1}^{n-n_1}I(X_{2,i};Y_{22,i}|U_{2,i},Z_{2,i})+\epsilon_n
\end{align}
Similarly, it can be shown that
\begin{align}
H(W_0,W_1,W_2|Z_1^{n_1},Z_2^{n-n_1})& \leq
\sum_{i=1}^{n_1}I(U_{1,i},Y_{21,i}|Z_{2,i})+\sum_{i=1}^{n_1}I(X_{1,i};Y_{11,i}|U_{1,i},Z_{1,i})\nonumber\\
&\quad +\sum_{i=1}^{n-n_1}I(X_{2,i};Y_{22,i}|Z_{2,i})
\end{align}
So far, we derived outer bounds on the secrecy capacity region
which match the achievable region. Hence, to claim that this is
indeed the capacity region, we need to show that computing the
outer bounds over all distributions of the form
$p(u_1,x_1)p(u_2,x_2)$ yields the same region which we would
obtain by computing over all $p(u_1,u_2,x_1,x_2)$. Since all the
expressions involved in the outer bounds depend on either
$p(u_1,x_1)$ or $p(u_2,x_2)$ but not on the joint distribution
$p(u_1,u_2,x_1,x_2)$, this argument follows, establishing the
secrecy capacity region.

\bibliographystyle{unsrt}
\bibliography{IEEEabrv,references2}
\end{document}